\documentclass{article}
\usepackage[utf8]{inputenc}
\usepackage{secdot, amsfonts, amsmath, amssymb, bm, breqn, mathtools, graphicx, derivative, caption, bbm, dsfont, float, subcaption, tikz, relsize, amsthm, fullpage, multicol, multirow, rotating, mathrsfs}
\usepackage{xcolor}
\usepackage{array}
\usepackage{authblk}
\usepackage[round]{natbib}
\setcitestyle{authoryear} 
\usepackage[title]{appendix}
\newtheorem{theorem}{Theorem}

\newtheorem{lemma}{Lemma}

\newtheorem{proposition}{Proposition}
\usepackage[export]{adjustbox} 
\usepackage[margin=0.8in]{geometry}
\usepackage[colorlinks=false, linktocpage=true]{hyperref}
\usepackage{cleveref}
\newcommand*{\vertbar}{\rule[-1ex]{0.5pt}{2.5ex}}

\hypersetup{
  colorlinks,
  citecolor=blue,
  linkcolor=blue,
  urlcolor=blue}
\allowdisplaybreaks 

\usepackage{setspace}
\begin{document}

\newcommand{\blind}{1}

\if1\blind
{
\title{Bayesian Flexible Modelling of Spatially Resolved Transcriptomic Data}
\author{\small Arhit Chakrabarti}
\author[1, a]{\small Yang Ni}
\author{\small Bani K. Mallick}
\affil{\footnotesize Department of Statistics, Texas A\&M University, College Station, TX}
\affil[a]{\footnotesize All correspondence should be addressed to \href{mailto:yni@stat.tamu.edu}{yni@stat.tamu.edu}}
\date{}

  \maketitle
} \fi

\if0\blind
{
  \bigskip
  \bigskip
  \bigskip
  \begin{center}
    {\LARGE\bf Bayesian Flexible Modelling of Spatially Resolved Transcriptomic Data}
\end{center}
  \medskip
} \fi

\begin{abstract}
\sloppy Single-cell RNA-sequencing technologies may provide valuable insights to the understanding of the composition of different cell types and their functions within a tissue. Recent technologies such as \emph{spatial transcriptomics}, enable the measurement of gene expressions at the single cell level along with the spatial locations of these cells in the tissue. Dimension-reduction and spatial clustering are two of the most common exploratory analysis strategies for spatial transcriptomic data. However, existing dimension reduction methods may lead to a loss of inherent dependency structure among genes at any spatial location in the tissue and hence do not provide insights of gene co-expression pattern. In spatial transcriptomics, the matrix-variate gene expression data, along with spatial co-ordinates of the single cells, provides information on both gene expression dependencies and cell spatial dependencies through its row and column covariances. In this work, we propose a flexible Bayesian approach to simultaneously estimate the row and column covariances for the matrix-variate spatial transcriptomic data. The posterior estimates of the row and column covariances provide data summaries for downstream exploratory analysis. We illustrate our method with simulations and two analyses of real data generated from a recent spatial transcriptomic platform. Our work elucidates gene co-expression networks as well as clear spatial clustering patterns of the cells. 
\end{abstract}
\bigskip
\sloppy \noindent%
\par\bigskip\bigskip
\newpage
\section{Introduction}
Single-cell RNA-sequencing technologies have been used to create molecular profiles for individual cells, which may provide valuable insights to the understanding of the composition of different cell types and their functions within a tissue. With newer technologies such as \emph{spatial transcriptomics}, it is now possible to measure gene expressions at the single cell level along with the information of spatial locations of these cells in the tissue. Such technologies include the earlier fluorescence in situ hybridization (FISH) based approaches (e.g., seqFISH \citep{seqFISH} and MERFISH \citep{MERFISH}), sequencing-based methods (e.g., 10x Visium \citep{10xVisium} and Slide-seq \citep{SlideSeq}), and the spatially-resolved transcript amplicon readout mapping (STARmap) \citep{STARmap}; see \citealp{ST_review} for a review of different spatial transcriptomic technologies. Spatial transcriptomic data  bring new scientific questions and statistical challenges to its analysis and interpretation. \par
A typical first step for processing genomics data is dimension reduction. There are numerous dimension-reduction techniques that have been routinely performed as a part of standard scRNA-seq data analysis pipeline, e.g., principal component analysis (PCA), weighted PCA (\citealp{WPCA}), t-distributed stochastic neighbor embedding (t-SNE,  \citealp{t-SNE}), and uniform manifold approximation and projection (UMAP, \citealp{umap}). These dimension-reduction techniques have also been applied in the context of spatial transcriptomic data. More specifically to spatial transcriptomic data, spatial clustering is one of the most common exploratory analysis strategies. Spatial clustering aims to use spatial transcriptomic information to cluster cells in the tissue into multiple spatial clusters, thereby segmenting the entire tissue into multiple tissue structures or domains. This segmentation of the tissue structure may aid in the understanding  of spatial and functional organization of the tissue. Common spatial clustering methods for spatial transcriptomic data include SpaGCN \citep{SpaGCN}, the hidden Markov random field model \citep{HMRF}, BayesSpace \citep{BayesSpace}, SpatialPCA \citep{SpatialPCA}, and SC-MEB \citep{SC-MEB}. The majority of the spatial clustering methods first involve a dimension reduction step on the expression matrix using some standard technique followed by spatial clustering of the estimated low-dimensional embeddings. In particular, PCA is routinely required in many software packages used for spatial transcriptomic analyses as a pre-processing step for downstream clustering analysis like \texttt{BayeSpace} \citep{BayesSpace}, \texttt{SpaGCN} \citep{SpaGCN}, \texttt{SC-MEB} \citep{SC-MEB}, \texttt{Seurat} \citep{seurat1}, etc. Most dimension reduction methods are not model-based and hence the low-dimensional embeddings are considered to be error free. Moreover, dimension reduction and spatial clustering methods separately optimize two different loss functions, which may not be ideal if considered sequentially. More recently, attempts have been made to simultaneously achieve dimension reduction and spatial clustering \citep{DRSC} to overcome the potential drawbacks of sequentially performing dimension reduction and spatial clustering. However, though convenient for computational purposes, dimension reduction techniques may lead to a loss of the inherent dependency structure among genes (e.g., co-expression) at any spatial location in the tissue. \par
In many spatial transcriptomic studies (e.g., STARmap), the expression data are collected on a moderate number of genes for a large number of single cells along with their spatial information in the tissue. In such cases, it may be of interest to understand the association between the observed gene set or some subset of target genes, along with spatial clustering of the single cells. The existing methods do not have such provision of understanding the genetic association. More concretely, the expression data  observed for a set of $N$ genes over a relatively large number $n$ of single cells, constitute a matrix of expression data. The expression data are also accompanied with the $n\times d$ spatial co-ordinates of the single cells, where the dimension $d = 2 \text { or } 3$ depends on the profiling method used. The matrix-variate spatial transcriptomic data provide information on both gene expression dependencies and cell spatial dependencies through the row and column covariances (correlations) of the matrix-variate data. Gaussian processes are commonly used to model spatial data, which typically involve the specification of spatial dependence in the form of a covariance matrix/kernel. Existing spatial covariance estimation methods ignore the dependency structure among the rows (genes in our case) of the matrix-variate data and often rely on a parametric assumption on the covariance kernel. 
In this paper, we propose a flexible Bayesian approach to simultaneously estimate the row and column covariances for the matrix-variate spatial transcriptomic data without fixing a parametric column covariance kernel or assuming the rows to be independent. Moreover, the proposed approach is computationally efficient for a large number of spatial locations (i.e., cells). \par
The proposed method takes as input the spatial gene expression matrix after standard log-normalization and the spatial coordinates of the single cells in the tissue. The output from our methodology gives the joint posterior estimates of both the row and column covariances for the matrix-variate spatial transcriptomic data. These posterior covariance  matrices are summaries of gene and cell dependencies and may be used for further downstream analyses. For example, the estimated column covariance matrix may be used for spatial clustering of the cells in the tissue whereas the estimated row covariance matrix may be used to construct a gene co-expression network. 
\par
The rest of the paper is organized as follows. Section \ref{sec:prelims} gives a brief overview of some preliminaries needed for the remainder of the paper. In Section \ref{sec:cholesky_MN}, we mention the theoretical backdrop of our proposed methodology introduced in Section \ref{sec:methology_single_sample}. We extend our methodology for the case when we have multiple independent spatial transcriptomic data on a set of common genes in Section \ref{sec:methology_multiple_sample}. Section \ref{sec:simulations} provides simulations to illustrate our method, comparing the performance of our proposed method with the existing method of Bayesian nonparametric spatial covariance estimation. Sections \ref{sec:real_data_STARmap_genes} presents two different analyses of a real spatial transcriptomic dataset collected from the STARmap platform \citep{STARmap}. The paper concludes with a brief discussion in Section \ref{sec:discussion}.

\section{Preliminaries}\label{sec:prelims}
\subsection{Kullback-Leibler divergence}
\sloppy The Kullback-Leibler (KL) divergence between two probability measures $P$ and $Q$ is defined as 
$\mathds{D}_{KL}(P\,\Vert\, Q) = \int \log\left(\frac{\mathrm{d} P} { \mathrm{d}Q}\right)$. 
Intuitively, if $P$ is the true data generating distribution and we consider a model with distribution $Q$, then the KL divergence is the expected loss of information in using the distribution $Q$ to model the true underlying distribution $P$. If $P$ and $Q$ are $n$-dimensional multivariate normal distributions with zero means and covariance matrices $\Sigma_1$ and $\Sigma_2$, respectively, then the KL divergence has a closed form,
\begin{equation}\label{eq:kln}
    \mathds{D}_{KL}(\mathcal{N}_n(0, \Sigma_1) \, \Vert \, \mathcal{N}_n(0, \Sigma_2)) = \frac{1}{2}\left[\text{tr}(\Sigma_2^{-1}\Sigma_1) + \log|\Sigma_2| - \log|\Sigma_1| - n\right],
\end{equation}
where $\text{tr}()$ denotes the trace of a matrix and $|\cdot|$ denotes the determinant of a matrix.
\subsection{Matrix-normal distribution}
The matrix-normal distribution is a generalization of multivariate normal distribution to random matrices. Let $\bm{Y}$ be an $N\times n$ random matrix. The random matrix $\bm{Y}$ is said to be distributed according to the matrix-normal distribution if its probability density function is given by,
\[ p(\bm{Y}\mid M, \Lambda, \Sigma) = \frac{\exp\left( -\frac{1}{2}\text{tr}[\Sigma^{-1}(\bm{Y} - M)^\top \Lambda^{-1}(\bm{Y} - M)
] \right)}{(2\pi)^{Nn/2} |\Sigma|^{N/2}  |\Lambda|^{n/2} }.\]
This is denoted as $\bm{Y}\sim \mathcal{M N}_{N, n} (M, \Lambda, \Sigma)$, where $M$ is the $N\times n$ mean matrix, $\Lambda$ is the $N \times N$ row covariance matrix, and $\Sigma$ is the $n \times n$ column covariance matrix. The matrix-normal distribution is related to a multivariate normal distribution as,
    \begin{equation}
    \label{eq:matrix_normal_equiv_def}
        \bm{Y}\sim \mathcal{M N}_{N, n} (M, \Lambda, \Sigma) \Leftrightarrow \text{vec}(\bm{Y}) \sim \mathcal{N}_{Nn}(\text{vec}(M), \Sigma\otimes \Lambda),
    \end{equation}
    where $\otimes$ denotes the Kronecker product and $\text{vec}(\cdot)$ denotes the vectorization of a matrix. The next lemma generalizes \eqref{eq:kln} to the KL divergence between two matrix-normal distributions.
\begin{lemma}
\label{KL_matrix_normal_lemma}
Consider two centered matrix-normal distributions  $P : \mathcal{MN}_{N,n}\left(0, \Lambda_1, \Sigma_1 \right)$ and $Q : \mathcal{MN}_{N,n}\left(0, \Lambda_2, \Sigma_2 \right)$. The KL divergence between $P$ and $Q$ is given by,
\begin{equation}
\label{KL_matrix_normal}
    \mathds{D}_{KL}(P\, \Vert \, Q) = \frac{1}{2} \left[\textup{tr}\left((\Sigma_2^{-1}\Sigma_1) \otimes (\Lambda_2^{-1}\Lambda_1)\right) - N \log\frac{|\Sigma_1|}{|\Sigma_2|} - n \log\frac{|\Lambda_1|}{|\Lambda_2|} - Nn \right].
\end{equation}
\end{lemma}
\begin{proof}
     Using the equivalent representation of matrix-normal distribution in \eqref{eq:matrix_normal_equiv_def} and from \eqref{eq:kln}, we have,
    \begin{align*}
        \mathds{D}_{KL}(P\, \Vert \, Q) & = \frac{1}{2} \left[\textup{tr}\left((\Sigma_2^{-1}\Sigma_1) \otimes (\Lambda_2^{-1}\Lambda_1)\right) + \log|\Sigma_2 \otimes \Lambda_2| - \log|\Sigma_1 \otimes \Lambda_1|  - Nn \right] \\
        & = \frac{1}{2} \left[\textup{tr}\left((\Sigma_2^{-1}\Sigma_1) \otimes (\Lambda_2^{-1}\Lambda_1)\right) + \log\left(|\Sigma_2|^N |\Lambda_2|^n\right) - \log\left(|\Sigma_1|^N |\Lambda_1|^n\right)  - Nn \right] \\
         & = \frac{1}{2} \left[\textup{tr}\left((\Sigma_2^{-1}\Sigma_1) \otimes (\Lambda_2^{-1}\Lambda_1)\right) - N \log\frac{|\Sigma_1|}{|\Sigma_2|} - n \log\frac{|\Lambda_1|}{|\Lambda_2|}  - Nn \right]
    \end{align*}
The second equality follows as $|A \otimes B| = \left(|A|^p |B|^q\right)$ for $A\in \mathds{R}^{q\times q}, B\in \mathds{R}^{p\times p}$, and the last equality follows from rearranging the terms.
\end{proof}
\section{Sparse approximate Cholesky factorization for matrix-normal distributions}\label{sec:cholesky_MN}
The precision matrix, also known as the inverse of the covariance matrix, encodes the partial correlations. For a random vector distributed as multivariate normal, the precision matrix is used to learn about the conditional independencies between the variables. If two variables are conditionally independent given all the other variables, the corresponding elements of the precision matrix are zero. Formally, for a random vector $\boldsymbol{Y} = (Y_1, \dots, Y_n) \sim \mathcal{N}_n(0, \Sigma)$, the precision matrix is given by $\Omega = \Sigma^{-1}$. Then, $\Omega_{ij} = 0 \Leftrightarrow Y_i \perp Y_j|\boldsymbol{Y}\setminus\{Y_i, Y_j\}$, where $\perp$ denotes the independence relationship. When the dimension $n$ is large, it is difficult to store the entire dense covariance matrix $\Sigma$ in memory and to compute its inverse (e.g., to evaluate the density). 
As such an inversion requires $\mathcal{O}(n^3)$ operations, many authors have proposed to consider approximate sparse Cholesky factorization of the precision matrix. Numerous algorithms have been proposed to improve the computational complexity of the Cholesky factorization.\par
We consider the case of matrix-normal distributions, $\mathcal{M N}_{N, n} (0, \Lambda, \Sigma)$. In many problems, the number $n$ of columns of the matrix-variate data (the number of cells in our case) is large. Similarly to the multivariate normal case described above, it becomes difficult to store the dense matrix $\Sigma$ in memory and the computation of $\Sigma^{-1}$ becomes infeasible. We propose an optimal approximate sparse Cholesky factor for the column precision matrix $\Sigma^{-1}$. Consider
\begin{equation}
\label{KL_minimization}
    L = \text{argmin}_{\hat{L} \in \mathscr{S}} \mathds{D}_{KL}\left( \mathcal{MN}_{N,n}\left(0, \Lambda, \Sigma \right) \Vert \mathcal{MN}_{N,n}\left(0, \Lambda, (\hat{L}\hat{L}^\top)^{-1} \right)\right)
\end{equation}
as an approximate Cholesky factor for $\Sigma^{-1}$ where $\mathscr{S}=\{A \in \mathds{R}^{n\times n} : A_{ij} \neq 0 \Leftrightarrow(i, j) \in S\}$ and $S\subset \{1,\dots, n\}\times \{1,\dots, n\}$ is the index set of the non-zero entries of a sparse lower-triangular matrix. Denoting by $\#$ the cardinality of a set, we characterize its property in the following theorem.
\begin{theorem}
\label{schafer_theorem_extension}
\sloppy The nonzero entries of the $i$th column of L as defined in \eqref{KL_minimization} are given by
\begin{equation}
    L_{s_i, i} =  \frac{\Sigma_{s_i, s_i}^{-1} \textbf{e}_1}{\sqrt{\textbf{e}_1^\top \Sigma_{s_i, s_i}^{-1} \textbf{e}_1}},
\end{equation}
where $s_i = \{j:(j,i)\in S\}$, $\Sigma_{s_i, s_i}$ is the restriction of $\Sigma$ to the set of indices $s_i$, $\Sigma_{s_i, s_i}^{-1} = (\Sigma_{s_i, s_i})^{-1}$, and $\mathbf{e}_1\in \mathds{R}^{\# s_i \times1}$ is the vector with the first entry equal to one and all other entries equal to zero. Using this formula, $L$ can be computed in computational
complexity $\mathcal{O}\left(\# S + ( \max_{1\leq i \leq n} \# s_i)^2\right)$ in space and $\mathcal{O}\left(\sum_{i=1}^{n} \# s_i^3\right)$ in time.
\end{theorem}
\begin{proof}
Note that using Lemma \ref{KL_matrix_normal_lemma},
\begin{align*}
    & \mathds{D}_{KL}\left( \mathcal{MN}_{N,n}\left(0, \Lambda, \Sigma \right) \Vert \mathcal{MN}_{N,n}\left(0, \Lambda, (\hat{L}\hat{L}^\top)^{-1} \right)\right)\\
    =\  & \frac{1}{2}\left[\text{tr}\left( \left(\hat{L}\hat{L}^\top\, \Sigma \right) \otimes \mathds{I}_N \right) - N  \log\frac{|\Sigma|}{|(\hat{L}\hat{L}^\top)^{-1}|} - N n\right]\\
    =\  & \frac{N}{2}\left[\text{tr}\left( \hat{L}\hat{L}^\top\, \Sigma \right) - \log\frac{|\Sigma|}{|(\hat{L}\hat{L}^\top)^{-1}|} - n\right]\\
    =\  & \frac{N}{2} \left[\ \text{tr}\left( \hat{L}\hat{L}^\top\, \Sigma \right) +  \log|(\hat{L}\hat{L}^\top)^{-1}| - \log|\Sigma| - n\ \right]\\
    = \ & N \  \mathds{D}_{KL}\left( \mathcal{N}_{n}\left(0, \Sigma \right) \Vert \,\mathcal{N}_{n}\left(0, (\hat{L}\hat{L}^\top)^{-1} \right)\right)
\end{align*}
Then, 
\begin{align*}
    L & = \text{argmin}_{\hat{L} \in \mathscr{S}} \mathds{D}_{KL}\left( \mathcal{MN}_{N,n}\left(0, \Lambda, \Sigma \right) \Vert \mathcal{MN}_{N,n}\left(0, \Lambda, (\hat{L}\hat{L}^\top)^{-1} \right)\right)\\
    & = \text{argmin}_{\hat{L} \in \mathscr{S}}\,N \, \mathds{D}_{KL}\left( \mathcal{N}_{n}\left(0, \Sigma \right) \Vert \,\mathcal{N}_{n}\left(0, (\hat{L}\hat{L}^\top)^{-1} \right)\right)
\end{align*}
The proof follows immediately from Theorem 2.1 of \citealp{schafer_KL}.
\end{proof}
Theorem \ref{schafer_theorem_extension} has the advantage of giving the best possible Cholesky factor of the column precision matrix as measured by KL divergence for a given sparsity pattern in a computationally efficient manner when $\# s_i, \forall \, i$ is chosen to be much smaller than $n$. For the case of a Gaussian process, the matrix $\Sigma$ is the covariance matrix generated from some kernel function $K(\,.\,,\,.\,)$ for a set of spatial locations $\{\boldsymbol{s}_i, \, i =1, \dots, n\}$. Specifically, the $(i, j)$th element of the covariance matrix is given by $\Sigma_{ij} = K(\boldsymbol{s}_i, \boldsymbol{s}_j)$ and the $j$th column of $\Sigma$ corresponds to the covariances of the $j$th spatial location $\boldsymbol{s}_j$ with $\{\boldsymbol{s}_i, \, i =1, \dots, n\}$.  \citealp{schafer_KL} provides several algorithms to order the rows and columns of the precision matrix corresponding to the spatial locations and selecting a sparsity pattern for the approximation. They further provide theoretical bounds for the computational complexity and approximation error of their proposed method. Theorem \ref{schafer_theorem_extension} essentially states that if we consider a ``maximum-minimum" ordering of the spatial locations, then we can efficiently get accurate sparse Cholesky factor of the column precision matrix.\par
Next, we theoretically show that a method ignoring the row dependence is sub-optimal in the KL sense when the true distribution has row dependence.
\begin{proposition}
\label{KL_better_lemma}
    Consider the centered matrix-normal distribution  $P : \mathcal{MN}_{N,n}\left(0, \Lambda, \Sigma \right)$, where both the column and row covariance matrices are positive definitive. Let $L$ be any $n\times n$ lower-triangular matrix. Denote by $Q : \mathcal{MN}_{N,n}\left(0, \mathds{I}_N, \left(LL^\top\right)^{-1} \right)$ and $R : \mathcal{MN}_{N,n}\left(0, \Lambda, \left(LL^\top\right)^{-1} \right)$. Further if $ \lambda_1 < \dots < \lambda_N$ denote the eigenvalues of $\Lambda$, then there exists a $\lambda^*$ and $\lambda^{**}$ depending on $L$ and $\Sigma$, such that for $\lambda_1 \geq \lambda^*$ or $\lambda_N \leq \lambda^{**}$, 
\begin{equation}
\label{KL_better}
    \mathds{D}_{KL}(P\, \Vert \, Q) \geq \mathds{D}_{KL}(P\, \Vert \, R),
\end{equation}
with equality holding if $\Lambda = \mathds{I}_N$.
\end{proposition}
\begin{proof}
    From Lemma \ref{KL_matrix_normal_lemma}, we have 
    \begin{align}
        \nonumber \mathds{D}_{KL}(P\, \Vert \, Q) & = \frac{1}{2}\left[\text{tr}\left( \left(LL^\top\Sigma\right )\otimes\Lambda\right) - N \log\frac{|\Sigma|}{|\left(LL^\top\right)^{-1}|} - n\log|\Lambda| - Nn \right] \\
        \nonumber & = \frac{1}{2}\left[N\left\{\text{tr}\left(LL^\top \Sigma\right) - \log\frac{|\Sigma|}{|\left(LL^\top\right)^{-1}|}  - n\right\} + (\text{tr}(\Lambda) - N)\text{tr}\left( LL^\top \Sigma\right) - n\log|\Lambda|\right]\\
        \label{KL_lemma_2_equation} & = \mathds{D}_{KL}(P\, \Vert \, R) + \frac{1}{2}\left[  (\text{tr}(\Lambda) - N)\text{tr}\left( LL^\top \Sigma\right) - n\log|\Lambda| \right]
    \end{align}
Note that $\epsilon = \text{tr}\left(LL^\top \Sigma\right) > 0$, as $\Sigma$ and $LL^\top$ are positive-definite. Since $0 <\lambda_1< \dots < \lambda_N$ are the eigenvalues of $\Lambda$,   
\begin{equation}
\label{eq:lambda_lemma_2}
    (\text{tr}(\Lambda) - N)\text{tr}\left( LL^\top \Sigma\right) - n\log|\Lambda| = \sum_{i = 1}^{N}\left\{(\lambda_i - 1)\epsilon - n\log\lambda_i\right\}
\end{equation}
It is easy to see that the function $f(\lambda) = (\lambda-1)\epsilon - n\log\lambda$ is convex and is minimized at $\lambda = n/\epsilon$. Further at the minimizer $n/\epsilon$, $f(\lambda) = 0$ if and only if $n/\epsilon = 1$. For $n/\epsilon \neq 1$, $f(\lambda) < 0$ at the minimizer. Since $f$ is convex, there exists a $\lambda^* \geq n/\epsilon$, i.e., a $\lambda^*$ depending on $L$ and $\Sigma$ such that for any $\lambda \geq \lambda^*$, $f(\lambda) \geq  0$. Similarly, there exists a $\lambda^{**} \leq n/\epsilon$ depending on $L$ and $\Sigma$ such that for any $\lambda \leq \lambda^{**}$, $f(\lambda) \geq  0$. The result follows immediately from \eqref{eq:lambda_lemma_2} and \eqref{KL_lemma_2_equation}. 
\end{proof}
Proposition \ref{KL_better_lemma} suggests that existing spatial methods such as \cite{brian},  which ignore the row dependence (i.e., setting $\Lambda = \mathds{I}_N$), may not perform well for spatial transcriptomic data where the rows (i.e., genes) are correlated.

\section{Covariance estimation for matrix-normal distributions}\label{sec:methology_single_sample}
Consider an $N\times n$ matrix $\bm{Y}$ of spatial transcriptomic data where $N$ denotes the number of genes and $n$ denotes the number of cells measured at the spatial locations $\boldsymbol{s}_1, \dots, \boldsymbol{s}_n$,
\begin{equation}
\label{data}
\boldsymbol{Y} = 
    \begin{pmatrix}
    y_1^{(1)} & \cdots & y_n^{(1)}\\
    \vdots & \ddots & \vdots\\
    y_1^{(N)} & \cdots & y_n^{(N)}
    \end{pmatrix} =  
    \begin{pmatrix}
    \vertbar &   & \vertbar\\
     \boldsymbol{y}_1& \cdots & \boldsymbol{y}_n\\
    \vertbar &  & \vertbar
    \end{pmatrix}.
\end{equation}
Here $y_i^{(\ell)}$ is the expression of the $\ell$th gene in the $i$th cell at location $\boldsymbol{s}_i$. We order the spatial locations $\boldsymbol{s}_1, \dots, \boldsymbol{s}_n$, and hence the columns of $\boldsymbol{Y}$ according to a maximin ordering \citep{guinness, schafer_approx}, which sequentially adds to the ordering the location that maximizes the minimum distance from the locations already in the ordering. We model $\boldsymbol{Y}$ as a centered matrix-normal distribution,
\begin{equation}
\label{model}
    \boldsymbol{Y} \sim \mathcal{MN}_{N, n}(0, \Lambda, \Sigma ),
\end{equation}
where $\Lambda$ and $\Sigma$ are the row and column covariance matrices. We focus on problems where the number of spatial locations $n$ is much larger than the number of genes $N$. We further consider an ordered conditional independence assumption,
\begin{equation}
\label{conditional independence}
    p(\boldsymbol{y}_i\mid \boldsymbol{y}_{1:i-1},\Lambda, \Sigma) = p(\boldsymbol{y}_i\mid \boldsymbol{y}_{g_m(i)}, \Lambda, \Sigma), \hspace{0.6cm} i=2, \dots, n,
\end{equation}
where $g_m(i) \subset \{1, \dots, i-1\}$ is an index vector consisting of the indices of the $\min(m, i-1)$ nearest neighbours to $\boldsymbol{s}_i$ among those ordered previously. Note that \eqref{conditional independence} holds trivially for $m = n-1$. Many authors have demonstrated both numerically and theoretically that \eqref{conditional independence} holds (at least approximately) even for $m\ll n$ for many covariance functions in the context of Vecchia approximations of parametric covariance functions (see e.g., \citealp{vecchia, stein, dutta, guinness, vecchia_GP, katzfuss_guinness, schafer_KL}). 

\subsection{Bayesian regression model framework}
Consider the modified Cholesky decomposition of the column precision matrix,
\begin{equation}
    \label{cholesky}
    \Sigma^{-1} = \text{U}\text{D}^{-1}\text{U}^\top,
\end{equation}
where $\text{D}=\text{diag}(d_1,\dots, d_n)$ is a diagonal matrix with positive entries $d_i >0$, and $\text{U}$ is a unit upper triangular matrix, i.e., an upper triangular matrix with diagonals equal to one. The ordered conditional independence in \eqref{conditional independence} implies that $\text{U}$ is sparse with at most $m$ nonzero off-diagonal elements per column. Theorem \ref{schafer_theorem_extension} ensures that the sparse approximate Cholesky factor of $\Sigma^{-1}$ minimizes the KL divergence of our proposed model from the true underlying matrix-variate distribution. Defining $\boldsymbol{u}_i = \text{U}_{g_m(i)}$ as the nonzero off-diagonal entries in the $i$th column of $\text{U}$, the model \eqref{model} can be written as a series of linear regression models \citep{huang}:
\begin{align}
\label{regression}
    \begin{aligned}
     p(\boldsymbol{Y}\mid \Lambda , \Sigma) & = \prod_{i=1}^{n} p(\boldsymbol{y}_i\mid \boldsymbol{y}_{g_m(i)}, \Lambda, \Sigma)\\
     & = \prod_{i=1}^{n} \mathcal{N}_N(\boldsymbol{y}_i\mid  \boldsymbol{X}_i\boldsymbol{u}_i, d_i\Lambda), \\
    \end{aligned}
\end{align}
where the ``design matrix" $\boldsymbol{X}_i$ consists of the observations at the $m$ neighboring locations of $\boldsymbol{s}_i$, stored in the columns of $\boldsymbol{Y}$ with indices $g_m(i)$, i.e. $\boldsymbol{X}_i$  is an $N\times m$ matrix with the $\ell$th row $-\boldsymbol{y}_{g_m(i)}^{(\ell)\top}$. Note that under this notation $\textbf{X}_1 = \boldsymbol{0}_{\footnotesize N\times1}$. For efficient Bayesian inference of the model parameters, we assign conjugate priors. For $i=1, \dots, n,$
\begin{alignat}{4}
\label{priors}
\begin{aligned}
     &\boldsymbol{u}_i\mid d_i &\overset{ind}{\sim} &\ \ \mathcal{N}_m(\boldsymbol{0}, d_i \boldsymbol{V}_i), \\
     & d_i &\overset{ind}{\sim} &\ \ \mathcal{IG}(\alpha_i, \beta_i), \\
     &\Lambda  &\overset{ind}{\sim} &\ \ \mathcal{IW}(\nu, \Psi).
    \end{aligned}
\end{alignat}
If the number $N$ of observed genes  is moderate, an Inverse-Wishart ($\mathcal{IW}$) prior on the row covariance matrix $\Lambda$ leads to closed form expressions for the full conditional distributions.  When $N$ is large, Bayesian latent factor models could be used instead \citep{factormodel_mikewest1, factormodel_mikewest2, factormodel_ABDD}. 
\subsection{Full conditional distributions}
In this section, we derive the full-conditional distributions of the model parameters $\boldsymbol{u} = \{\boldsymbol{u}_1, \dots, \boldsymbol{u}_n\}$, $\boldsymbol{d} = \{d_1, \dots, d_n\}$, and $\Lambda$. Specifically, the full-conditional distribution of $(\boldsymbol{u}, \boldsymbol{d})$ is given by,
\begin{align}\label{full_conditional_u_d}
    \nonumber p(\boldsymbol{u}, \boldsymbol{d} \mid \boldsymbol{Y}, \Lambda) & \propto p(\boldsymbol{Y}\mid \boldsymbol{u}, \boldsymbol{d}, \Lambda )\, p(\boldsymbol{u}, \boldsymbol{d})\\
    & \propto \prod_{i = 1}^{n} d_i^{-\left(\frac{m + N}{2} + \alpha_i + 1\right)} \exp\left(-\frac{1}{d_i}\left(\beta_i + \frac{1}{2}u_i^\top \boldsymbol{V}_i^{-1} u_i + \frac{1}{2}(\boldsymbol{y}_i - \boldsymbol{X}_i\boldsymbol{u}_i)^\top \Lambda^{-1} (\boldsymbol{y}_i - \boldsymbol{X}_i\boldsymbol{u}_i)\right) \right).
\end{align}
Simplifying the exponent in \eqref{full_conditional_u_d}, we have
\begin{multline*}
        p(\boldsymbol{u}, \boldsymbol{d} \mid \boldsymbol{Y}, \Lambda) \propto \prod_{i=1}^{n} d_i^{-\left(\frac{N}{2} + \alpha_i + 1 + \frac{m}{2}\right)} \exp\left(-\frac{1}{d_i}\left(\beta_i + \frac{1}{2}\boldsymbol{y}_i^\top(\Lambda^{-1} - \Lambda^{-1}\boldsymbol{X}_i \boldsymbol{G}_i^{-1}\boldsymbol{X}_i^\top\Lambda^{-1})\boldsymbol{y}_i\right)\right)\\
        \times \exp\left( -\frac{1}{2d_i} (\boldsymbol{u}_i - \boldsymbol{G}_i^{-1}\boldsymbol{H}_i)^\top \boldsymbol{G}_i (\boldsymbol{u}_i - \boldsymbol{G}_i^{-1}\boldsymbol{H}_i)  \right).
\end{multline*}
Therefore,
\begin{equation}
    \label{full_conditiona_NIG}
    p(\boldsymbol{u}, \boldsymbol{d} \mid \boldsymbol{Y}, \Lambda) \equiv \prod_{i=1}^{n} \mathcal{N}_m\left(\boldsymbol{u}_i \mid \boldsymbol{G}_i^{-1} \boldsymbol{H}_i, d_i\boldsymbol{G}_i^{-1}\right)\mathcal{IG}(d_i\mid \tilde{\alpha}_i, \tilde{\beta}_i),
\end{equation}
where 
\begin{align}
\label{eq:forms_of_posterior_parameters}
    \begin{aligned}
        & \tilde{\alpha}_i = \alpha_i + \frac{N}{2}, \\
        & \tilde{\beta}_i  = \beta_i + \frac{1}{2}\boldsymbol{y}_i^\top\left(\Lambda^{-1} - \Lambda^{-1}\boldsymbol{X}_i\boldsymbol{G}_i^{-1} \boldsymbol{X}_i^\top \Lambda^{-1}\right) \boldsymbol{y}_i,\\
        & \boldsymbol{G}_i^{-1} = \left(\boldsymbol{V}_i^{-1} + \boldsymbol{X}_i^\top \Lambda^{-1} \boldsymbol{X}_i\right)^{-1},\\
        & \boldsymbol{H}_i = \boldsymbol{X}_i^\top \Lambda^{-1} \boldsymbol{y}_i.
    \end{aligned}
\end{align}
\noindent The full conditional distribution of $\Lambda$ is given by,
\begin{align*}
    \nonumber p(\Lambda \mid \boldsymbol{Y}, \boldsymbol{u}, \boldsymbol{d}) & \propto p(\boldsymbol{Y}\mid \boldsymbol{u}, \boldsymbol{d}, \Lambda )\, p(\Lambda)\\
    \nonumber & \propto \left\{|\Lambda|^{-\frac{n}{2}} \prod_{i=1}^{n} d_i^{-\frac{N}{2}}\exp\left( -\frac{1}{2d_i}\text{tr}\left((\boldsymbol{y}_i - \boldsymbol{X}_i\boldsymbol{u}_i)(\boldsymbol{y}_i - \boldsymbol{X}_i\boldsymbol{u}_i)^\top \Lambda^{-1}\right)\right) \right\}|\Lambda|^{-\frac{N+\nu + 1}{2}}\exp\left(-\frac{1}{2}\text{tr}(\Psi \Lambda^{-1})\right)\\
    & \propto |\Lambda|^{-\frac{N+ (\nu + n) + 1}{2}} \exp\left(-\frac{1}{2}\text{tr}\left(\left(\Psi + \sum_{i=1}^{n}\frac{(\boldsymbol{y}_i - \boldsymbol{X}_i\boldsymbol{u}_i)(\boldsymbol{y}_i - \boldsymbol{X}_i\boldsymbol{u}_i)^\top}{d_i}\right)\Lambda^{-1}\right)\right).
\end{align*}
Hence, 
\begin{equation}\label{full_condition_lambda}
    p(\Lambda \mid \boldsymbol{Y}, \boldsymbol{u}, \boldsymbol{d})  \equiv \mathcal{IW}\left(\Psi + \sum_{i=1}^{n}\frac{(\boldsymbol{y}_i - \boldsymbol{X}_i\boldsymbol{u}_i)(\boldsymbol{y}_i - \boldsymbol{X}_i\boldsymbol{u}_i)^\top}{d_i}\,,\, n + \nu\right),
\end{equation}
Because \eqref{full_conditiona_NIG} and \eqref{full_condition_lambda} are in closed form, Gibbs sampling is straightforward. 
\subsection{Parameterization and inference on the hyperparameters}\label{subsec:hyperparmeter}
We reparameterize the priors for $\bm{u}_i$ and $d_i$ in \eqref{priors} in terms of a much smaller
number of hyperparameters. Inspired by the behavior of Mat\'ern-type
covariance functions, we introduce a three-dimensional vector of hyperparameters $\boldsymbol{\theta} = (\theta_1, \theta_2, \theta_3)^\top$, where $\theta_1$ is related to the marginal variance, $\theta_2$ is related to the range, and $\theta_3$ is related to the smoothness. 
We first consider the prior for $d_i$ in \eqref{priors}. Considering an exponential covariance kernel with marginal variance $\theta_1$ and range $2/\theta_2$, the $(i,j)$th entry of the spatial covariance matrix, $\Sigma_{ij} = \theta_1\exp(-\theta_2\Vert\boldsymbol{s}_i -\boldsymbol{s}_j\Vert/2)$. Assuming $m = 1$, we have,
\begin{align}
    \nonumber \text{Var}(\boldsymbol{y}_i\mid \boldsymbol{y}_{g_m(i)}, \Lambda) & = d_i \Lambda\\
    \label{decay_functional_form}
     & = \theta_1(1 - \exp(-\theta_2\Vert\boldsymbol{s}_i -\boldsymbol{s}_g\Vert))\Lambda \\
    \label{approx_decay_functional_form}
    & \approx \theta_1(1 - \exp(-\theta_2(i)^{-\frac{1}{p}})\Lambda, 
\end{align}
where $g = g_1(i)$. The functional form in \eqref{decay_functional_form} holds exactly for the exponential covariance kernel with $m = 1$ and holds at least approximately for Mate\'rn covariance kernels in two dimensions with $m = n-1$. The approximation in \eqref{approx_decay_functional_form} follows since under a maximin ordering of the spatial locations, the distance $\Vert\boldsymbol{s}_i -\boldsymbol{s}_g\Vert$ between the location $\bm{s}_i$ and its nearest previously ordered neighbor decreases roughly as $(i)^{-1/p}$ for a regular grid on a unit hypercube $[0,1]^p$. The prior mean and variance of $d_i\Lambda$ conditional on $\Lambda$ is given by,
\begin{align}
\label{prior_mean_d}
    \mathds{E}(d_i\Lambda) & = \frac{\beta_i}{\alpha_i - 1}\Lambda,\\
    \label{prior_variance_d}
    \text{Var}(d_i\Lambda) & = \frac{\beta_i^2}{(\alpha_i - 1)^2 (\alpha_i - 2)}\Lambda^2.
\end{align}  
With the motivation of getting a prior for $d_i$ that shrinks toward \eqref{approx_decay_functional_form},
we set the prior mean \eqref{prior_mean_d} to be equal to the functional form in \eqref{approx_decay_functional_form}. Thus, for $i =1, \dots, n$,
\begin{align}
    \nonumber \frac{\beta_i}{\alpha_i - 1}\Lambda & = \theta_1(1 - \exp(-\theta_2(i)^{-\frac{1}{p}})\Lambda \\
    \label{eq:alpha_beta1}
    \Rightarrow \frac{\beta_i}{\alpha_i - 1} & = \theta_1 f_{\theta_2}(i),
\end{align}
where $f_{\theta_2}(i) = (1 - \exp(-\theta_2(i)^{-\frac{1}{p}})$. As the empirically observed variance of $d_i$ decreases with the index $i$ as well, following this reasoning, we set the prior standard deviation of $d_i \Lambda$ (obtained from \eqref{prior_variance_d}) to be half the prior mean in \eqref{prior_mean_d}. Therefore, for $i=1,\dots, n$, 
\begin{align}
    \nonumber \frac{\beta_i}{(\alpha_i - 1)\sqrt{\alpha_i -2}}\Lambda & = \frac{\beta_i}{2(\alpha_i - 1)}\Lambda\\
    \label{eq:alpha_beta2}\Rightarrow \frac{\beta_i}{(\alpha_i - 1)\sqrt{\alpha_i -2}} & = \frac{\beta_i}{2(\alpha_i - 1)}.
\end{align}
Solving for $\alpha_i$ and $\beta_i$ from the equations \eqref{eq:alpha_beta1} and \eqref{eq:alpha_beta2} yields,  for $i = 1, \dots, n$,
\begin{equation}\label{prior_parameters_form1}
\begin{aligned}
    \alpha_i & = 6, \\
    \beta_i & = 5 \theta_1 f_{\theta_2}(i).
    \end{aligned}
\end{equation}
Recent results based on elliptic boundary-value problems \citep{schafer_approx} imply that the Cholesky entry $(\boldsymbol{u}_i)_j$, corresponding to the $j$th nearest neighbor, decays exponentially as a function of $j$ for Mate\'rn covariance functions whose spectral densities are the reciprocal of a polynomial (ignoring edge effects). With the motivation to capture this exponential decay for the Cholesky entries, we arrive at the same functional form for the entries $v_{ij}$ of the diagonal matrix $\boldsymbol{V}_i$ in \eqref{priors}.
Specifically, we have
\begin{equation}\label{prior_parameters_form2}
\bm{V}_i = \text{Diag}(v_{i1}, \dots, v_{im}),\hspace{0.2cm} \text{where $v_{ij} = \exp\left(-\frac{\theta_3\,j}{f_{\theta_2}(i)}\right),\ \  j = 1, \dots, m$}, \  i = 1, \dots, n.
\end{equation}
We note that all components of $\boldsymbol{\theta}$ are assumed to be positive and so we perform all inference on the logarithmic scale. We next discuss how to infer the hyperparameter $\boldsymbol{\theta}$ based on the data $\bm{Y}$. The key component for learning about the hyperparameter $\boldsymbol{\theta} = (\theta_1, \theta_2, \theta_3)^\top$ is the marginal or integrated likelihood $p(\boldsymbol{Y}\mid \Lambda, \boldsymbol{\theta})$, 
\begin{align}
    p(\boldsymbol{Y}\mid \Lambda, \boldsymbol{\theta}) 
    & = \prod_{i = 1}^{n} \int_{d_i} \int_{\boldsymbol{u}_i} \mathcal{N}_N(\boldsymbol{y}_i \mid \boldsymbol{X}_i\boldsymbol{u}_i, d_i \Lambda)\, \mathcal{N}_m(\boldsymbol{u}_i\mid \boldsymbol{0}, d_i\boldsymbol{V}_i)\, \mathcal{IG}(d_i \mid \alpha_i, \beta_i) \, \text{d}d_i\text{d}\boldsymbol{u}_i.
\end{align}
Standard calculations on multivariate normal distribution yield,
\begin{equation*}
    p(\boldsymbol{Y}\mid \Lambda, \boldsymbol{\theta}) \propto \prod_{i=1}^{n} \frac{|\boldsymbol{G}_i^{-1}|^{\frac{1}{2}}}{|\boldsymbol{V}_i|^{\frac{1}{2}}} \frac{\beta_i^{\alpha_i}}{\tilde{\beta}_i^{\tilde{\alpha}_i}} \frac{\Gamma(\tilde{\alpha}_i)}{\Gamma(\alpha_i)},
\end{equation*}
where the prior parameters $\alpha_i, \beta_i, \bm{V}_i$ are given in \eqref{priors}, and the parameters $\tilde{\alpha}_i, \tilde{\beta}_i, \boldsymbol{G}_i^{-1}$ are given in \eqref{eq:forms_of_posterior_parameters}.
For a  fully Bayesian inference, we assume a flat prior for $\boldsymbol{\theta}$, and the marginal posterior distribution of $\boldsymbol{\theta}$ is given by,
\begin{align}
    \nonumber p(\boldsymbol{\theta}\mid \Lambda, \boldsymbol{Y}) & \propto p(\boldsymbol{Y}\mid \Lambda, \boldsymbol{\theta})\\
    \label{full_conditiona_theta}
    & \propto \prod_{i=1}^{n} \frac{|\boldsymbol{G}_i^{-1}|^{\frac{1}{2}}}{|\boldsymbol{V}_i|^{\frac{1}{2}}} \frac{\beta_i^{\alpha_i}}{\tilde{\beta}_i^{\tilde{\alpha}_i}} \frac{\Gamma(\tilde{\alpha}_i)}{\Gamma(\alpha_i)}. 
\end{align}
\subsection{Blocked Gibbs sampling algorithm}
From the full conditional distributions  \eqref{full_conditiona_NIG}, \eqref{full_condition_lambda}, and the marginal posterior distribution \eqref{full_conditiona_theta}, the blocked Gibbs sampling is straightforward. The algorithm iterates by sampling from the conditional distributions,
\begin{align}
    \begin{aligned}
        p(\boldsymbol{\theta}\mid \Lambda, \boldsymbol{Y}) & \propto \prod_{i=1}^{n} \frac{|\boldsymbol{G}_i^{-1}|^{\frac{1}{2}}}{|\boldsymbol{V}_i|^{\frac{1}{2}}} \frac{\beta_i^{\alpha_i}}{\tilde{\beta}_i^{\tilde{\alpha}_i}} \frac{\Gamma(\tilde{\alpha}_i)}{\Gamma(\alpha_i)}, \\
        p(\boldsymbol{u}, \boldsymbol{d} \mid \boldsymbol{Y}, \Lambda, \boldsymbol{\theta}) & \propto \prod_{i=1}^{n} \mathcal{N}_m\left(\boldsymbol{u}_i \mid \boldsymbol{G}_i^{-1} \boldsymbol{H}_i, d_i\boldsymbol{G}_i^{-1}\right)\mathcal{IG}(d_i\mid \tilde{\alpha}_i, \tilde{\beta}_i), \\
        p(\Lambda \mid \boldsymbol{Y}, \boldsymbol{u}, \boldsymbol{d}, \boldsymbol{\theta})  & \propto \mathcal{IW}\left(\Psi + \sum_{i=1}^{n}\frac{(\boldsymbol{y}_i - \boldsymbol{X}_i\boldsymbol{u}_i)(\boldsymbol{y}_i - \boldsymbol{X}_i\boldsymbol{u}_i)^\top}{d_i}\,,\, n + \nu\right),
    \end{aligned}
\end{align}
where the parameters $\tilde{\alpha}_i, \tilde{\beta}_i, \bm{H}_i$ and $\bm{G}_i^{-1}$ are given by \eqref{eq:forms_of_posterior_parameters}. The functional dependencies of the prior parameters $\alpha_i, \beta_i,\bm{V}_i$ on the hyperparameter $\bm{\theta}$ are specified by \eqref{prior_parameters_form1} and \eqref{prior_parameters_form2}.

\section{Covariance estimation for matrix-normal distributions with multiple independent samples}\label{sec:methology_multiple_sample}
In many cases, we have independent samples of spatial transcriptomic data measured on the same set of genes. For example, the experiment may collect spatially resolved single-cell gene expression data for a set of genes of interest from a number of experimental units (e.g, different tissue samples). The data from the $r$th sample $\boldsymbol{Y}_r$ is an $ N \times n_r$ matrix, where $n_r$ denotes the number of single cells observed for the $r$th sample and $N$ denotes the number of genes. The spatial locations of the single cells $\boldsymbol{s}_{r1}, \dots, \boldsymbol{s}_{r,n_r}$ may not align for different samples $r,\,  r= 1, \dots, R$. We consider the same maximin ordering of the spatial locations corresponding to each sample $\boldsymbol{Y}_r$. Then each $\boldsymbol{Y}_r$ is modelled independently as a centered matrix-normal distribution with a shared row covariance matrix but a sample-specific column covariance matrix,
\[
\boldsymbol{Y}_r \overset{ind}{\sim}\mathcal{MN}_{N, n_r} (0, \Lambda, \Sigma_r), \ \ r = 1, \dots, R.
\]
Similarly as before, we take the modified Cholesky decomposition of the column precision matrix for each sample,
\begin{equation}\label{eq:cholesky_each_sample}
    \Sigma_r^{-1} = \text{U}_r\text{D}_r^{-1}\text{U}_r^\top.
\end{equation}
Letting $\boldsymbol{Y} = \{\boldsymbol{Y}_1, \dots, \boldsymbol{Y}_R\}$ denote the collection of all samples, we have a similar representation of the joint distribution of $\bm{Y}$ in terms of series of linear regression models as \eqref{regression}, 
\begin{align}
\label{regression_replicates}
    \begin{aligned}
     p(\boldsymbol{Y}\mid \Lambda , \{\Sigma_1,\dots, \Sigma_R\}) & = \prod_{r=1}^{R}\prod_{i=1}^{n_r} p(\boldsymbol{y}_{ri}\mid \boldsymbol{y}_{rg_{r,m}(i)}, \Lambda, \Sigma_r)\\
     & = \prod_{r=1}^{R}\prod_{i=1}^{n_r} \mathcal{N}_N(\boldsymbol{y}_{ri}\mid  \boldsymbol{X}_{ri}\boldsymbol{u}_{ri}, d_{ri}\Lambda), \\
    \end{aligned}
\end{align}
where the ``design matrix" $\boldsymbol{X}_{ri}$ of the $r$th sample consists of the observations at the $m$ neighboring locations of $\boldsymbol{s}_{ri}$, stored in the columns of $\boldsymbol{Y}_r$ with indices $g_{r,m(i)}$.
Similarly, $\boldsymbol{u}_{ri} = \text{U}_{r, g_{r,m(i)}}$ is the nonzero off-diagonal entries in the $i$th column of $\text{U}_r$, and $d_{ri}$ is the $i$th diagonal element of the diagonal matrix $\text{D}_r$ in \eqref{eq:cholesky_each_sample}.  
We assume independent priors that are conjugate to model \eqref{regression_replicates}. For $i=1, \dots, n_r,\ r=1, \dots, R,$
\begin{alignat}{4}
\label{priors_replicates}
\begin{aligned}
     &\boldsymbol{u}_{ri}\mid d_{ri}&\overset{ind}{\sim} &\ \ \mathcal{N}_m(\boldsymbol{0}, d_{ri} \boldsymbol{V}_{ri}),\\
     & d_{ri} &\overset{ind}{\sim} &\ \ \mathcal{IG}(\alpha_{ri}, \beta_{ri})\\
     &\Lambda  &\overset{ind}{\sim} &\ \ \mathcal{IW}(\nu, \Psi).
    \end{aligned}
\end{alignat}
Similarly to Section \ref{subsec:hyperparmeter}, we reparameterize the priors for $\bm{u}_{ri}$ and $d_{ri}$ in terms of a shared vector of hyperparameters $\boldsymbol{\theta} = (\theta_1, \theta_2, \theta_3)^\top$.  The blocked Gibbs sampling algorithm for the multi-sample case follows as an immediate extension to the corresponding single-sample algorithm. Letting $\boldsymbol{u} = \{\boldsymbol{u}_{ri}, i = 1, \dots, n_r, \ r = 1, \dots, R\}$ and $\boldsymbol{d} = \{d_{ri}, i = 1, \dots, n_r, \ r = 1, \dots, R\}$, the algorithm iterates by sampling,
\begin{align}
    \begin{aligned}
        p(\boldsymbol{\theta}\mid \Lambda, \boldsymbol{Y}) & \propto \prod_{r=1}^{R} \prod_{i=1}^{n_r} \frac{|\boldsymbol{G}_{ri}^{-1}|^{\frac{1}{2}}}{|\boldsymbol{V}_{ri}|^{\frac{1}{2}}} \frac{\beta_{ri}^{\alpha_{ri}}}{\tilde{\beta}_{ri}^{\tilde{\alpha}_{ri}}} \frac{\Gamma(\tilde{\alpha}_{ri})}{\Gamma(\alpha_{ri})}, \\
        p(\boldsymbol{u}, \boldsymbol{d} \mid \boldsymbol{Y}, \Lambda, \boldsymbol{\theta}) & \propto \prod_{r=1}^{R} \prod_{i=1}^{n_r}  \mathcal{N}_m\left(\boldsymbol{u}_{ri} \mid \boldsymbol{G}_{ri}^{-1} \boldsymbol{H}_{ri}, d_{ri}\boldsymbol{G}_{ri}^{-1}\right)\mathcal{IG}(d_{ri}\mid \tilde{\alpha}_{ri}, \tilde{\beta}_{ri}), \\
        p(\Lambda \mid \boldsymbol{Y}, \boldsymbol{u}, \boldsymbol{d}, \boldsymbol{\theta})  & \propto \mathcal{IW}\left(\Psi + \sum_{r=1}^{R}\sum_{i=1}^{n_r}\frac{(\boldsymbol{y}_{ri} - \boldsymbol{X}_{ri}\boldsymbol{u}_{ri})(\boldsymbol{y}_{ri} - \boldsymbol{X}_{ri}\boldsymbol{u}_{ri})^\top}{d_{ri}}\,,\, \sum_{r=1}^{R} n_r + \nu\right),
    \end{aligned}
\end{align}
where, for $i = 1, \dots, n_r$, $r=1, \dots, R$,
\begin{align}
    \begin{aligned}
        \alpha_{ri} & = 6,\\
        \tilde{\alpha}_{ri} & = \alpha_{ri} + \frac{N}{2},\\
        \beta_{ri} & = 5\theta_1(1 - \exp(-\theta_2(i)^{-\frac{1}{p}}),\\
        \tilde{\beta}_{ri} & = \beta_{ri} + \frac{1}{2}\boldsymbol{y}_{ri}^\top\left(\Lambda^{-1} - \Lambda^{-1}\boldsymbol{X}_{ri}\boldsymbol{G}_{ri}^{-1} \boldsymbol{X}_{ri}^\top \Lambda^{-1}\right) \boldsymbol{y}_{ri}, \\
        \boldsymbol{G}_{ri}^{-1} & = \left(\boldsymbol{V}_{ri}^{-1} + \boldsymbol{X}_{ri}^\top \Lambda^{-1} \boldsymbol{X}_{ri}\right)^{-1},\\
        \boldsymbol{H}_{ri} & = \boldsymbol{X}_{ri}^\top \Lambda^{-1} \boldsymbol{y}_{ri},\\
        \boldsymbol{V}_{ri} & = \text{Diag}(v_{ri1}, \dots, v_{rim}),\\
        v_{rij} & = \exp\left(-\frac{\theta_3\,j}{(1 - \exp(-\theta_2(i)^{-\frac{1}{p}})}\right),\ \  j = 1, \dots, m.\\
    \end{aligned}
\end{align}
\section{Simulations}\label{sec:simulations}
In this section we provide simulations to demonstrate the performance of the proposed method in the estimation of row and column covariance matrices for a matrix-normal distribution. We compare our method with the existing Bayesian nonparametric method of spatial covariance estimation for multivariate data \citep{brian}. We consider the case of a single sample of matrix-variate data in Section \ref{sec:simulations_single_sample} and the multi-sample case in Section \ref{sec:simulations_multiple_sample}. 
\subsection{Single-sample case}\label{sec:simulations_single_sample}
We considered $n = 100, 200, 500$ spatial locations and $N =20, 30$ genes. The outline of the data generation and simulation strategies are as follows:
\begin{itemize}
    \item We drew $n$ spatial locations from $\mathcal{U}(0, 1)\times \mathcal{U}(0, 1)$, where $\mathcal{U}(0, 1)$ denotes a uniform distribution on the interval $(0,1)$.
    \item The true column covariance $\Sigma$ was generated from a Mat\'ern covariance kernel with smoothness parameter equal to $0.25$, marginal variance equal to $1$, and varied the range parameter ($\phi = 1, 2$), using the $n$ random spatial locations generated.
    \item The true row covariance $\Lambda$ was generated from $\mathcal{IW}(N, \Psi)$ where the scale matrix $\Psi$ was considered to be one of the following choices:
        \begin{table}[H]
        \centering
            \begin{tabular}{ |c|c|c| } 
                    \hline
                      AR-correlation & Equi-correlation & Banded-correlation \\ 
                     \hline
                     $\Psi_{ij}=\rho^{|i - j|}$ & $\Psi_{ij}=\rho^{\mathds{1}(i\neq j)}$ & $\Psi_{ij}=\rho^{1 - \mathds{1}(i=j)}\mathds{1}(|i - j| \in \{0, 1\})$\\ 
                     \hline
            \end{tabular}
            \caption{Choices of scale matrix $\Psi$ for generating the true $\Lambda$. For all cases, $\rho = 0.5$.}
    \label{tab:scale_matrix_choices}
        \end{table}
        \item We fixed $m = 10$ for all our simulations.
        \item We used an \emph{adaptive-MCMC} algorithm for sampling the hyperparameters $\boldsymbol{\theta}= (\theta_1, \theta_2, \theta_3)^\top$, using the publicly available R package \texttt{adaptMCMC}. We initialized $(\theta_1, \theta_2, \theta_3) = (1, -1, 0)$. 
 The initial scale (shape) matrix required for the proposal distribution of the \emph{adaptive-MCMC} algorithm was taken as 
$
 \begin{bmatrix}
\phantom{-}0.05 &           -0.04 &    0 \\
          -0.04 & \phantom{-}0.05 &    0 \\
             0  &               0 & 0.01
 \end{bmatrix}
 $
        \item We ran $2,000$ iterations of our sampler and discarded the first $1,000$ samples as burn-in. 
        \item To monitor convergence of our MCMC, we plotted the traceplots of log likelihood.
        \item Since for the matrix-normal distribution, the covariance matrices are non-identifiable, we consider the posterior correlation matrices and carry out comparisons based on the correlation matrix rather than the covariance matrix. We still denote by $\Sigma$, $\Lambda$, etc. the correlation matrices rather than the covariance matrix.
        \item We performed our simulations for 30 independent replicates for any given combination of simulation parameters.
\end{itemize}
We compared $\text{KL}_{\mathcal{N}}$ and $\text{KL}_{\mathcal{MN}}$, which denote the KL divergence (in log scale) by assuming that $\boldsymbol{Y} \sim \mathcal{MN}_{N, n}(0, \mathds{I}, \Sigma )$ and $\boldsymbol{Y} \sim \mathcal{MN}_{N, n}(0, \Lambda, \Sigma )$, respectively. Note that $\mathcal{MN}_{N, n}(0, \mathds{I}, \Sigma )$ corresponds to $N$ independent realizations of a multivariate normal $\mathcal{N}_n(0, \Sigma)$. We compared both these quantities from our posterior estimates of $\Sigma$ and/or $\Lambda$, denoting them with the subscript ``P" and the MAP estimate of $\Sigma$ obtained from the method by \citealp{brian}, denoting them with subscript ``M". Hereafter, we refer to the method by \citealp{brian} as the MAP method. We define the relative Frobenius error of an estimator $\hat{A}$ of the matrix $A$ by,
\[
    RE_{\hat{A}} = \frac{\Vert \hat{A} - A \Vert_F}{\Vert A \Vert_F},
\]
where $F$ denotes the Frobenius norm of a matrix. We denote by $RE_{\Sigma_{P}}$, $RE_{\Lambda_{P}}$, and $RE_{\Sigma_{M}}$ the relative Frobenius error of the correlation matrix $\Sigma$ and $\Lambda$, with the subscript denoting the method used. Note that $\Lambda$ is not estimated under MAP method. Table \ref{tab:KL_comparisons_replicates} summarizes the results of our simulations wherein we report the mean and standard deviation over the replicates for the different comparison metrics. It is clear that in situations when the rows of a matrix-variate data (genes) are correlated, the KL divergences and relative Frobenius errors are lower for our proposed method than that under the MAP method. We also note that as the number of spatial locations increases, the accuracy of estimation of the row correlations increases as can be seen from the corresponding decreasing relative Frobenius error.

\begin{sidewaystable}[htp]
        \centering
            \begin{tabular}{ |c|c|c|c|c|c|c|c|c|c|c|c } 
                    \hline
                    $\phi$ & N & $n$ & Correlation &$\text{KL}_{\mathcal{N}\,(\text{P})}$  & $\text{KL}_{\mathcal{N}\,(\text{M})}$  & $\text{KL}_{\mathcal{MN}\, (\text{P})}$ & $\text{KL}_{\mathcal{MN}\, (\text{M})}$ & $RE_{\Sigma_{\text{P}}}$ & $RE_{\Lambda_{\text{P}}}$ & $RE_{\Sigma_{\text{M}}}$\\ 
                     \hline
                   \multirow{18}*{1 } & \multirow{9}*{20}  &  100 & AR & 4.365 (0.184)& 7.196 (0.362) & 7.575 (0.178) & 10.246 (0.348) & 0.65 (0.099) & 0.241 (0.061) & 1.102 (0.156)\\
                    & & 200 & AR & 5.27 (0.117) & 8.158 (0.321) & 8.401 (0.115) & 11.195 (0.309) & 0.624 (0.072) & 0.169 (0.036) & 1.093 (0.169)\\
                    & & 500 & AR & 6.596 (0.122) & 9.651 (0.308) & 9.652 (0.132) & 12.67 (0.301) & 0.663 (0.089) & 0.112 (0.022) & 1.148 (0.191)\\
                    \cline{3-11}
                    & & 100 & Equi & 4.264 (0.142) & 6.988 (0.339) & 7.444 (0.194) & 10.049 (0.318) & 0.656 (0.093) & 0.228 (0.047) & 1.043 (0.146)\\
                    & & 200 & Equi & 5.319 (0.144) & 8.095 (0.272) & 8.428 (0.16) & 11.134 (0.261) & 0.613 (0.07) & 0.156 (0.032) & 1.015 (0.136)\\
                    & & 500 & Equi & 6.645 (0.137) & 9.453 (0.386) & 9.712 (0.151) & 12.478 (0.376) & 0.615 (0.089) & 0.114 (0.023) & 1.009 (0.145)\\
                    \cline{3-11}
                    & & 100 & Banded & 4.356 (0.215) & 7.735 (0.349) & 7.628 (0.223) & 10.77 (0.337) & 0.633 (0.11) & 0.24 (0.064) & 1.104 (0.179)\\
                    & & 200 & Banded & 5.276 (0.176) & 9.152 (0.573) & 8.441 (0.187) & 12.168 (0.565) & 0.659 (0.108) & 0.168 (0.037) & 1.145 (0.138)\\
                    & & 500 & Banded & 6.592 (0.142) & 10.567 (0.322) & 9.673 (0.169) & 13.574 (0.319) & 0.65 (0.101) & 0.111 (0.029) & 1.176 (0.164)\\
                     \cline{2-11}
                      & \multirow{9}*{30}  &  100 & AR & 4.253 (0.107) & 9.357 (0.379) & 8.019 (0.171) & 12.773 (0.374) & 0.642 (0.101) & 0.114 (0.036) & 1.471 (0.194)\\
                    & & 200 & AR & 5.284 (0.151) & 10.439 (0.346) & 8.892 (0.163) & 13.85 (0.343) & 0.596 (0.104) & 0.081 (0.023) & 1.452 (0.186)\\
                    & & 500 & AR & 6.615 (0.126) & 11.713 (0.377) & 10.104 (0.134) & 15.121 (0.374) & 0.588 (0.071) & 0.052 (0.015) & 1.544 (0.209)\\
                    \cline{3-11}
                    & & 100 & Equi & 4.207 (0.167) & 9.299 (0.293) & 7.952 (0.215) & 12.715 (0.289) & 0.647 (0.09) & 0.131 (0.032) & 1.387 (0.15)\\
                    & & 200 & Equi & 5.288 (0.157) & 10.256 (0.418) & 8.877 (0.171) & 13.67 (0.413) & 0.616 (0.085) & 0.093 (0.029) & 1.428 (0.17)\\
                    & & 500 & Equi & 6.623 (0.137) & 11.502 (0.416) & 10.096 (0.137) & 14.912 (0.413) & 0.583 (0.067) & 0.058 (0.016) & 1.485 (0.144)\\
                    \cline{3-11}
                    & & 100 & Banded & 4.339 (0.225) & 9.426 (0.278) & 8.167 (0.228) & 12.842 (0.274) & 0.636 (0.108) & 0.119 (0.04) & 1.494 (0.221)\\
                    & & 200 & Banded & 5.264 (0.147) & 10.851 (0.429) & 8.929 (0.202) & 14.259 (0.426) & 0.612 (0.091) & 0.084 (0.022) & 1.525 (0.176)\\
                    & & 500 & Banded & 6.544 (0.08) & 12.21 (0.443) & 10.09 (0.141) & 15.616 (0.442) & 0.615 (0.062) & 0.054 (0.016) & 1.585 (0.168)\\
                    \hline
                    \multirow{18}*{2 } & \multirow{9}*{20}  &  100 & AR & 4.316 (0.245) & 7.131 (0.449) & 7.53 (0.242) & 10.186 (0.427) & 0.523 (0.13) & 0.264 (0.064) & 0.65 (0.12)\\
                    & & 200 & AR & 5.256 (0.226) & 8.167 (0.44) & 8.395 (0.232) & 11.205 (0.423) & 0.517 (0.123) & 0.193 (0.058) & 0.692 (0.151)\\
                    & & 500 & AR & 6.592 (0.128) & 9.463 (0.45) & 9.654 (0.141) & 12.488 (0.439) & 0.49 (0.099) & 0.126 (0.043) & 0.698 (0.142)\\
                    \cline{3-11}
                    & & 100 & Equi & 4.237 (0.207) & 7.067 (0.377) & 7.425 (0.235) & 10.124 (0.357) & 0.53 (0.112) & 0.263 (0.083) & 0.676 (0.158)\\
                    & & 200 & Equi & 5.257 (0.179) & 8.171 (0.465) & 8.383 (0.173) & 11.209 (0.449) & 0.511 (0.113) & 0.203 (0.069) & 0.617 (0.114)\\
                    & & 500 & Equi & 6.622 (0.183) & 9.458 (0.364) & 9.675 (0.182) & 12.482 (0.356) & 0.485 (0.118) & 0.116 (0.03) & 0.667 (0.139)\\
                    \cline{3-11}
                    & & 100 & Banded & 4.347 (0.231) & 7.638 (0.424) & 7.624 (0.216) & 10.678 (0.405) & 0.492 (0.121) & 0.261 (0.066) & 0.698 (0.141)\\
                    & & 200 & Banded & 5.311 (0.225) & 9.027 (0.51) & 8.458 (0.191) & 12.046 (0.501) & 0.514 (0.131) & 0.174 (0.042) & 0.683 (0.131)\\
                    & & 500 & Banded & 6.559 (0.154) & 10.615 (0.657) & 9.609 (0.175) & 13.623 (0.652) & 0.531 (0.107) & 0.114 (0.028) & 0.829 (0.191)\\
                    \cline{2-11}
                      & \multirow{9}*{30}  &  100 & AR & 4.311 (0.249) & 9.295 (0.468) & 8.042 (0.284) & 12.713 (0.462) & 0.472 (0.14) & 0.138 (0.038) & 1.048 (0.226)\\
                    & & 200 & AR & 5.329 (0.208) & 10.362 (0.485) & 8.91 (0.212) & 13.775 (0.481) & 0.431 (0.098) & 0.097 (0.021) & 0.999 (0.216)\\
                    & & 500 & AR & 6.651 (0.13) & 11.521 (0.44) & 10.122 (0.136) & 14.931 (0.437) & 0.406 (0.074) & 0.053 (0.017) & 1.096 (0.208)\\
                    \cline{3-11}
                    & & 100 & Equi & 4.21 (0.19) & 9.362 (0.575) & 7.922 (0.243) & 12.779 (0.568) & 0.514 (0.117) & 0.137 (0.052) & 0.989 (0.254)\\
                    & & 200 & Equi & 5.254 (0.199) & 10.099 (0.351) & 8.831 (0.21) & 13.514 (0.347) & 0.467 (0.106) & 0.097 (0.043) & 1.006 (0.234)\\
                    & & 500 & Equi & 6.622 (0.151) & 11.749 (0.515) & 10.121 (0.172) & 15.157 (0.512) & 0.424 (0.103) & 0.061 (0.022) & 0.916 (0.199)\\ 
                    \cline{3-11}
                    & & 100 & Banded & 4.419 (0.305) & 9.626 (0.511) & 8.222 (0.285) & 13.04 (0.506) & 0.443 (0.125) & 0.134 (0.04) & 1.011 (0.266)\\
                    & & 200 & Banded & 5.262 (0.216) & 10.64 (0.412) & 8.921 (0.196) & 14.05 (0.409) & 0.49 (0.125) & 0.093 (0.024) & 1.018 (0.225)\\
                    & & 500 & Banded & 6.595 (0.163) & 12.24 (0.489) & 10.081 (0.181) & 15.646 (0.487) & 0.447 (0.104) & 0.052 (0.012) & 1.082 (0.218)\\
                    \hline
            \end{tabular}
            \caption{KL divergences (in log scale) and Relative Frobenius error for the two methods, with varying number of genes, number of spatial locations, and range of Mate\'rn kernel using estimated correlation matrices for the different choices of the scale matrix for the true row correlation (lower values indicate better fit). We report the mean (s.d) over $30$ independent replicates. }
            \label{tab:KL_comparisons_replicates}
\end{sidewaystable}  

\subsection{Multi-sample case}\label{sec:simulations_multiple_sample}
We further conducted our simulations for multiple independent samples. In particular we considered $3$ independent samples, i.e. $R = 3$ of spatial data on the same set of genes (N) over possibly different spatial locations. For simplicity, we considered that the three samples have the same number of spatial locations i.e., $n_1 = n_2 = n_3 =n$. As before, we considered $n = 100, 200, 500$, and $N =20, 30$. The outline of the data generation and
simulation strategies are as follows:
\begin{itemize}
    \item We drew $n_r,\  r = 1, 2, 3$ spatial locations independently from $\mathcal{U}(0, 1)\times \mathcal{U}(0, 1)$.
    \item The true column covariance $\Sigma_1$ for the sample 1 was generated from a Mat\'ern covariance kernel with smoothness parameter equal to $0.5$, marginal variance equal to $1$, and range parameter equal to $1$, using the $n_1$ random spatial locations generated.
    \item The true column covariance $\Sigma_2$ for the sample 2  was generated from a Mat\'ern covariance kernel with smoothness parameter equal to $0.5$, marginal variance equal to $1.5$, and range parameter equal to $1$, using the $n_2$ random spatial locations generated.
    \item The true column covariance $\Sigma_3$ for the sample 3 was generated from a Mat\'ern covariance kernel with smoothness parameter equal to $0.5$, marginal variance equal to $2$, and range parameter equal to $1$, using the $n_3$ random spatial locations generated.
    \item The true row covariance was generated from $\mathcal{IW}(N, \Psi)$, where the scale matrix $\Psi$ was considered to be one of the choices given in Table \ref{tab:scale_matrix_choices}.
    \item We performed our simulations for 30 independent replicates for any given combination of simulation parameters.
    \item All other data generation strategies, sampling parameters, and convergence criterion were similar to the single-sample case.
\end{itemize}
We calculated the relative Frobenius errors of the three spatial column correlation matrices and the row correlation matrix. From Table \ref{tab:RFN_comparisons_replicates}, it can be seen that the relative Frobenius error from our posterior estimates of spatial correlation matrices are smaller than that from the MAP estimates. Moreover, the relative Frobenius error from our posterior estimates decreases as $N$ increases under any row correlation structure, whereas it shows an increasing trend for the MAP estimates. Besides, as the number of spatial locations increases, the relative Frobenius error  of the estimated row correlations decreases for the proposed method. 

\begin{sidewaystable}[htp]
        \centering
            \begin{tabular}{ |c|c|c|c|c|c|c|c|c|c|c } 
                    \hline
                     $N$ & $n_r$ & Correlation & $RE_{\Sigma_{1(\text{P})}}$  & $RE_{\Sigma_{2(\text{P})}}$   & $RE_{\Sigma_{3(\text{P})}}$  & $RE_{\Lambda_{\text{P}}}$ & $RE_{\Sigma_{1(\text{M})}}$ & $RE_{\Sigma_{2(\text{M})}}$ & $RE_{\Sigma_{3(\text{M})}}$\\ 
                     \hline
                     \multirow{9}*{20}  &  100 & AR & 0.331 (0.023) & 0.335 (0.04) & 0.331 (0.026) & 0.134 (0.031) & 0.574 (0.129) & 0.594 (0.15) & 0.545 (0.127)\\
                     & 200 & AR & 0.338 (0.034) & 0.347 (0.028) & 0.37 (0.051) & 0.085 (0.014) & 0.616 (0.148) & 0.602 (0.144) & 0.597 (0.135)\\
                     & 500 & AR & 0.354 (0.046) & 0.347 (0.046) & 0.352 (0.029) & 0.054 (0.009) & 0.601 (0.128) & 0.637 (0.203) & 0.62 (0.163)\\
                    \cline{2-10}
                     & 100 & Equi & 0.343 (0.042) & 0.33 (0.038) & 0.341 (0.032) & 0.114 (0.024) & 0.564 (0.136) & 0.571 (0.128) & 0.584 (0.139)\\
                     & 200 & Equi & 0.337 (0.034) & 0.356 (0.03) & 0.355 (0.027) & 0.086 (0.016) & 0.592 (0.132) & 0.593 (0.158) & 0.559 (0.104)\\
                     & 500 & Equi & 0.348 (0.041) & 0.341 (0.028) & 0.357 (0.043) & 0.055 (0.011) & 0.573 (0.129) & 0.599 (0.176) & 0.606 (0.158)\\
                    \cline{2-10}
                     & 100 & Banded & 0.337 (0.033) & 0.328 (0.021) & 0.331 (0.026) & 0.128 (0.019) & 0.6 (0.144) & 0.582 (0.164) & 0.604 (0.162)\\
                     & 200 & Banded & 0.324 (0.021) & 0.345 (0.013) & 0.352 (0.023) & 0.092 (0.018) & 0.602 (0.17) & 0.643 (0.152) & 0.619 (0.195)\\
                     & 500 & Banded & 0.345 (0.031) & 0.344 (0.035) & 0.347 (0.025) & 0.059 (0.013) & 0.579 (0.107) & 0.621 (0.178) & 0.559 (0.113)\\
                     \cline{1-10}
                     \multirow{9}*{30}  &  100 & AR & 0.328 (0.048) & 0.323 (0.029) & 0.325 (0.027) & 0.064 (0.015) & 0.94 (0.263) & 0.932 (0.262) & 0.886 (0.253)\\
                     & 200 & AR & 0.33 (0.029) & 0.342 (0.017) & 0.351 (0.033) & 0.047 (0.014) & 0.863 (0.284) & 0.926 (0.285) & 0.945 (0.255)\\
                     & 500 & AR & 0.342 (0.029) & 0.33 (0.017) & 0.345 (0.024) & 0.029 (0.006) & 0.892 (0.267) & 0.928 (0.255) & 0.95 (0.242)\\
                    \cline{2-10}
                     & 100 & Equi & 0.319 (0.024) & 0.316 (0.022) & 0.325 (0.021) & 0.068 (0.018) & 0.82 (0.226) & 0.869 (0.241) & 0.922 (0.242)\\
                     & 200 & Equi & 0.324 (0.023) & 0.341 (0.02) & 0.346 (0.021) & 0.043 (0.008) & 0.86 (0.258) & 0.905 (0.234) & 0.842 (0.256)\\
                     & 500 & Equi & 0.339 (0.018) & 0.338 (0.02) & 0.344 (0.014) & 0.03 (0.007) & 0.829 (0.249) & 0.817 (0.231) & 0.826 (0.243)\\
                    \cline{2-10}
                     & 100 & Banded & 0.326 (0.021) & 0.32 (0.018) & 0.33 (0.023) & 0.067 (0.018) & 0.796 (0.244) & 0.808 (0.245) & 0.887 (0.28)\\
                     & 200 & Banded & 0.323 (0.02) & 0.338 (0.018) & 0.345 (0.014) & 0.044 (0.012) & 0.768 (0.211) & 0.854 (0.207) & 0.925 (0.247)\\
                     & 500 & Banded & 0.335 (0.029) & 0.333 (0.029) & 0.345 (0.03) & 0.027 (0.008) & 0.827 (0.224) & 0.88 (0.254) & 0.904 (0.237)\\
                    \hline
            \end{tabular}
            \caption{Relative Frobenius error for each spatial correlation matrix and the row correlation matrix norm for the two methods. Results are presented with varying number of genes, number of spatial locations, and different choices of the scale matrix for the true row correlation (lower values indicate better fit). Note that under MAP, the row correlation is not estimated. We report the mean (s.d) over $30$ independent replicates. }
            \label{tab:RFN_comparisons_replicates}
\end{sidewaystable}
\section{Real Data Analyses}\label{sec:real_data_STARmap_genes}
\subsection{One-sample analysis}
 We considered the STARmap (spatially-resolved transcript amplicon readout mapping) dataset \citep{STARmap}. We first analyzed one of the experimental mice and later extended the analysis to multiple mice. The experimental mouse was dark housed for four days and exposed to light for one hour before obtaining measurements from the primary visual cortex, which consisted of the expression of 160 genes in 975 single cells. The spatial locations of these single cells in the tissue were also recorded. We removed cells showing extreme expression of genes. The data were log-normalized with a scaling factor equal to the median expression of total reads per cell. We focused our analysis on the representative genes corresponding to ``excitatory", ``inhibitory", and ``non-neuronal" cell types, which consists of the $13$ genes as mentioned in \citealp{STARmap}. \par
 We ran the blocked Gibbs sampler for $5,000$ iterations. The traceplot (Figure \ref{fig:traceplot}) of the log-likelihood  after discarding the first $2500$ samples as burn-in showed no lack of convergence.

The autocorrelation plot (Figure \ref{fig:ACF}) of the log-likelihood showed low autocorrelation and hence no thinning was needed.

\begin{figure}[http]
\centering
\begin{subfigure}{0.65\textwidth}
  \centering
  \includegraphics[width= 1\linewidth]{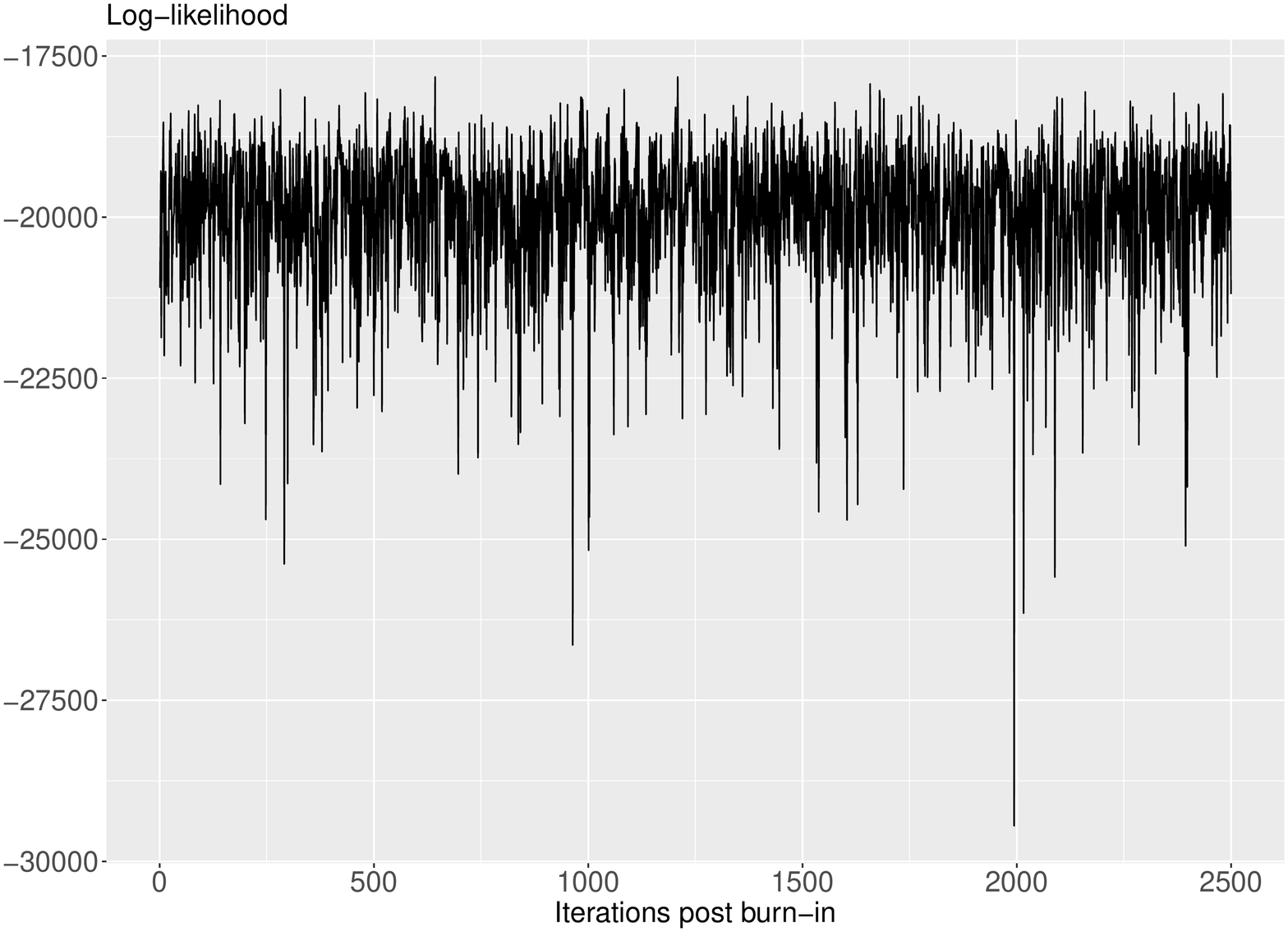}
  \caption{Traceplot of log-likelihood.}
  \label{fig:traceplot}
\end{subfigure}
\par
\begin{subfigure}{.65\textwidth}
  \centering
 \includegraphics[width= 1\linewidth]{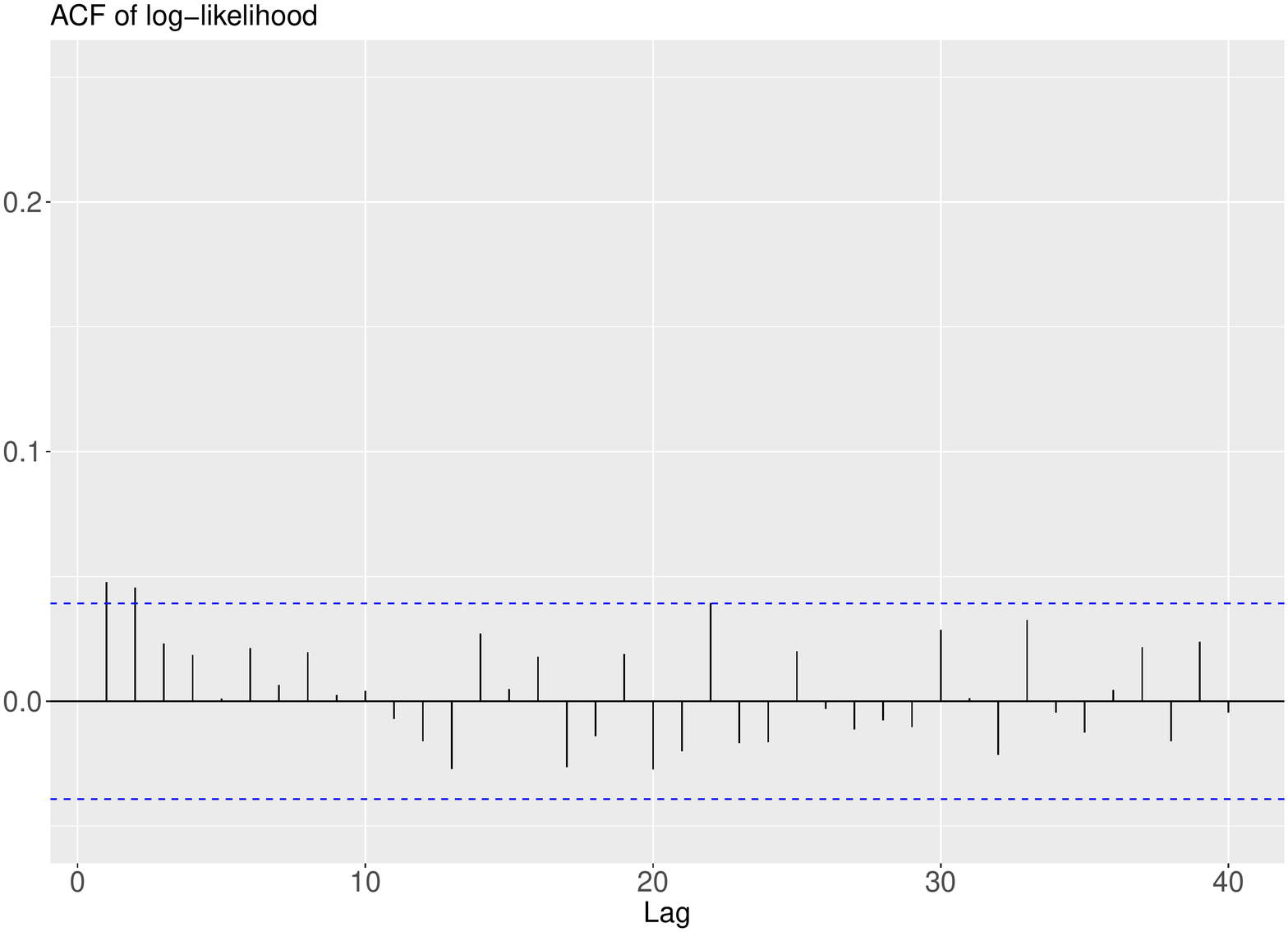}
  \caption{Autocorrelation plot of log-likelihood.}
  \label{fig:ACF}
\end{subfigure}
\caption{Single-sample. Traceplot and Autocorrelation plot of log-likelihood.}
\label{fig:traceplot_ACF}
\end{figure}
The posterior estimate of the row precision matrix (inverse covariance) was used to calculate the partial correlations among the representative genes. The heatmap of the partial correlations is shown in Figure \ref{fig:gene_corr}. 

 Further, considering only the magnitude of partial correlations, and a cut-off of $0.1$, we obtained the estimated co-expression network among the representative genes (Figure \ref{fig:gene_network}). The estimated co-expression network shows that ``Slc17a7" is a hub gene. The importance of this gene is established from its spatial expression pattern by the STARmap platform. Excitatory cells are directly related to the spreading of network activity in and outside of the neuronal networks in the brain \citep{excitatory_inhibitory}. This possibly validates the co-expression network between almost all the representative genes corresponding to ``excitatory cells" as well as association with some genes representative to ``inhibitory" and ``non-neuronal" cells.  The genes ``Gad1" and ``Pvalb", both being representative genes corresponding to ``inhibitory" cells, are connected in our estimated co-expression network, while the other genes corresponding to ``inhibitory cells" do not show co-expression pattern. This lack of association possibly conforms to the biological functioning of ``inhibitory" cells, which are more likely to inhibit stimuli to other cells \citep{inhibitory_functions}.

\begin{figure}[http]
\centering
\begin{subfigure}[t]{0.65\textwidth}
  \centering
  \includegraphics[width=1\linewidth]{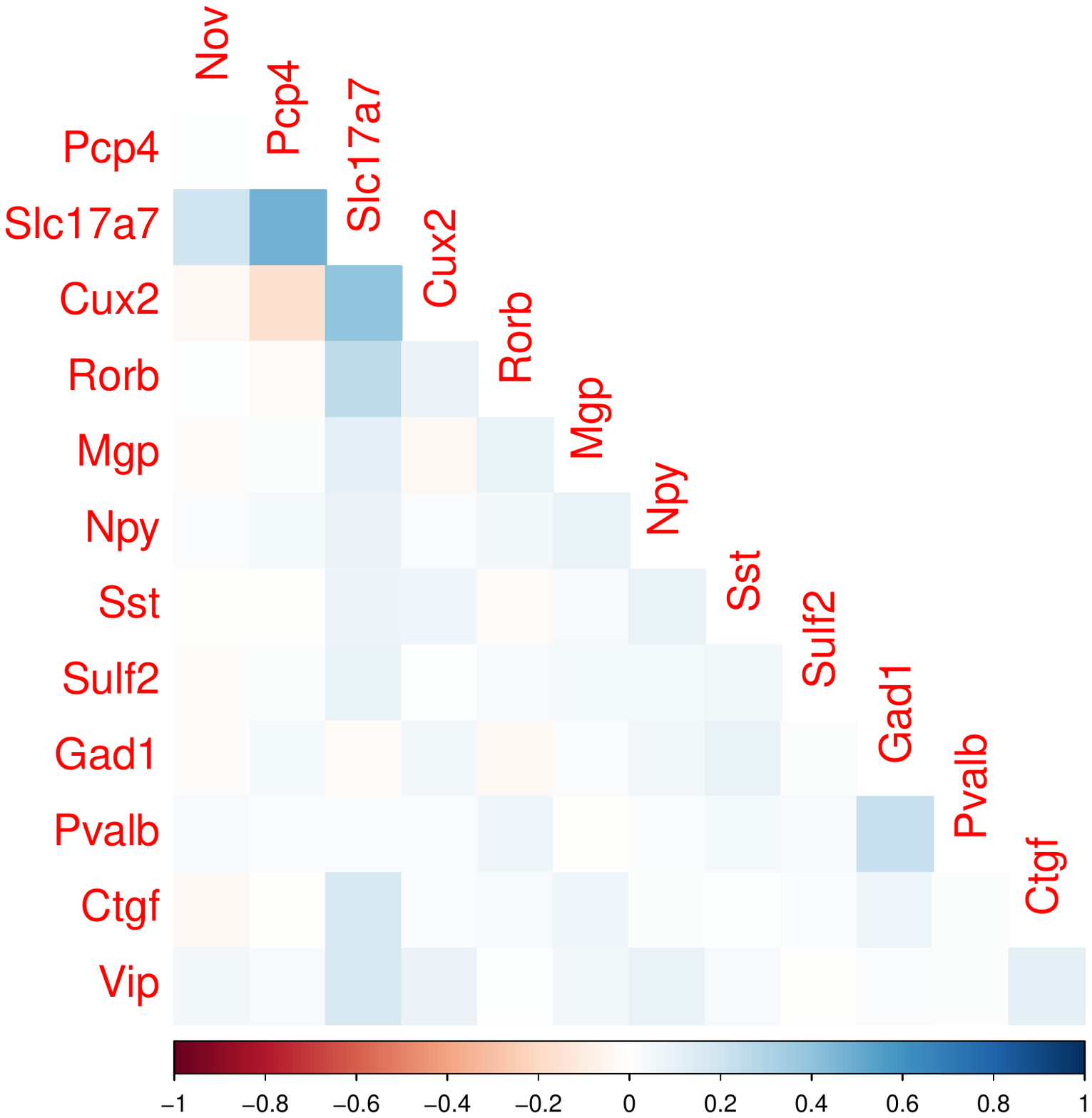}
  \caption{Heatmap of the estimate of partial correlations between the representative genes.}
  \label{fig:gene_corr}
\end{subfigure}
\qquad
\begin{subfigure}[t]{.65\textwidth}
  \centering
  \includegraphics[width= 1\linewidth]{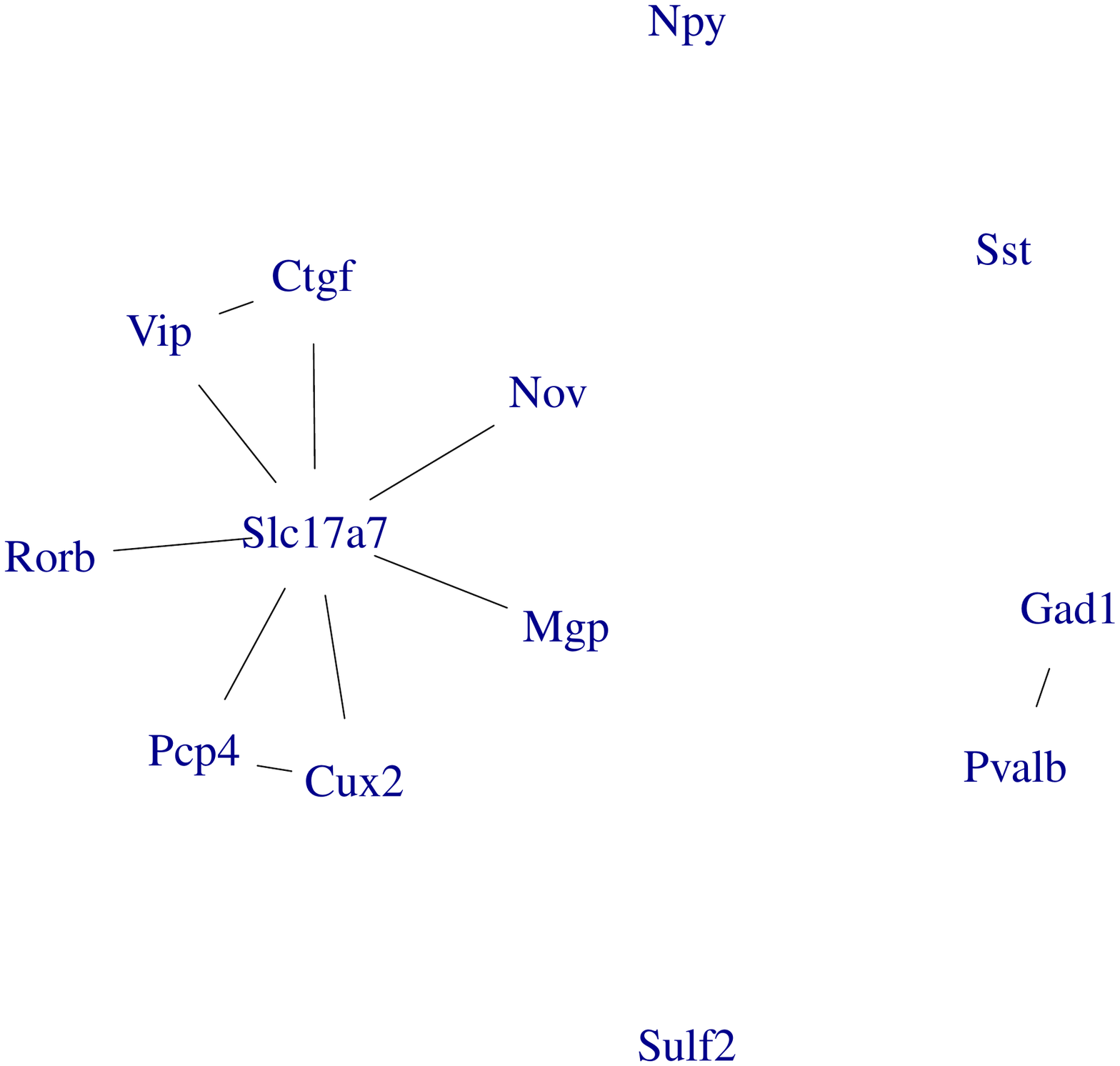}
  \caption{Estimated network among the representative genes.}
  \label{fig:gene_network}
\end{subfigure}
\caption{Single-sample analysis.}
\label{fig:gene_corr_network}
\end{figure}

\begin{figure}[http]
\centering
\begin{subfigure}{1\textwidth}
  \centering
  \includegraphics[width= 0.75\linewidth]{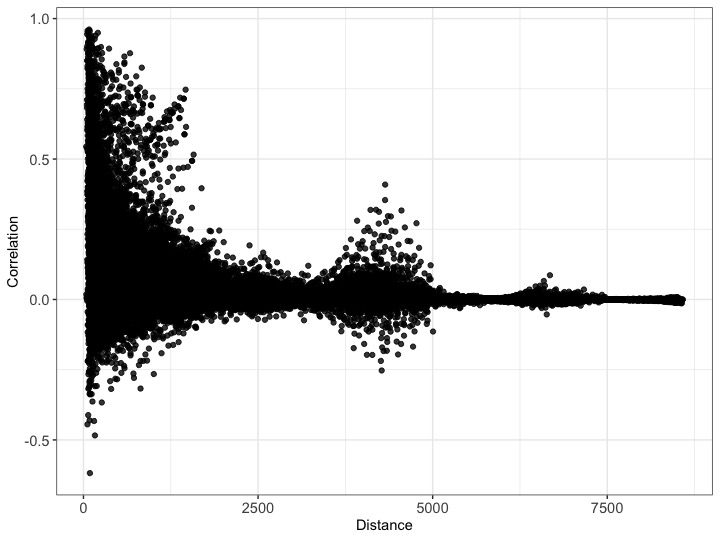}
  \caption{Spatial correlations estimated using our proposed method.}
   \label{fig:Spatial_Corr_Dist_Post}
\end{subfigure}
\par\bigskip
\begin{subfigure}{0.75\textwidth}
\centering
  \includegraphics[width= 1\linewidth]{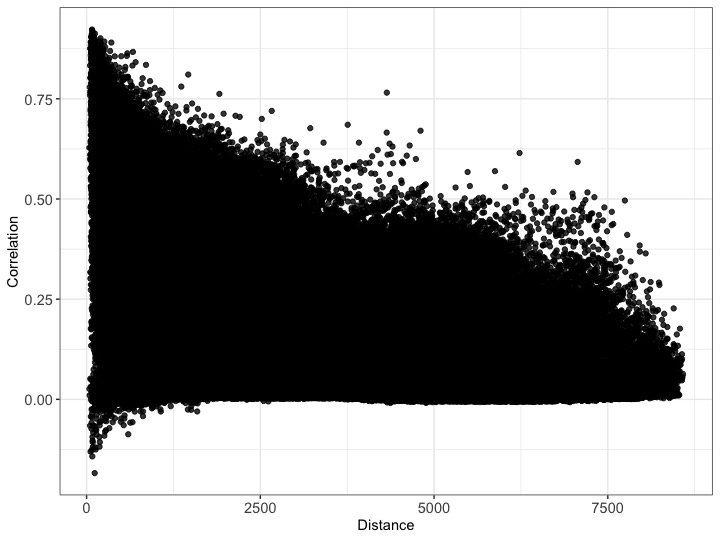} 
  \caption{Spatial correlations estimated from the MAP method.}
  \label{fig:Spatial_Corr_Dist_MAP}
\end{subfigure}
\caption{Single-sample. Plot of spatial correlations estimated from the two different methods against spatial distances between the pair of cells.}
\label{fig:Spatial_Corr_Dist}
\end{figure}
In Figure \ref{fig:Spatial_Corr_Dist}, we also plot the column spatial correlations against the corresponding distances for the proposed method as well as the MAP estimate. From the Figure \ref{fig:Spatial_Corr_Dist_Post}, we see that spatial correlation generally decays as the distance increases. However, the spatial correlation from the MAP estimate (Figure \ref{fig:Spatial_Corr_Dist_MAP}) does not decay significantly with distance. As the MAP model 
does not account for the correlation between genes, the spatial correlations are in general over-estimated using their method. From Figure \ref{fig:Spatial_Corr_Dist_Post}, we notice that there is a bump in the correlations whose corresponding distances are between $\sim 3,000$ and $\sim5,000$. Further investigation reveals that although the cells are spatially separated, the presence of correlated genes leads to this higher spatial correlations. The Figure \ref{fig:CellCellConnection}
shows the spatial plot of the cells exhibiting high correlations despite  being spatially separated. They may be of the same molecular cell type although they are far apart.

\begin{figure}[http]
\centering
\includegraphics[width =0.75\linewidth]{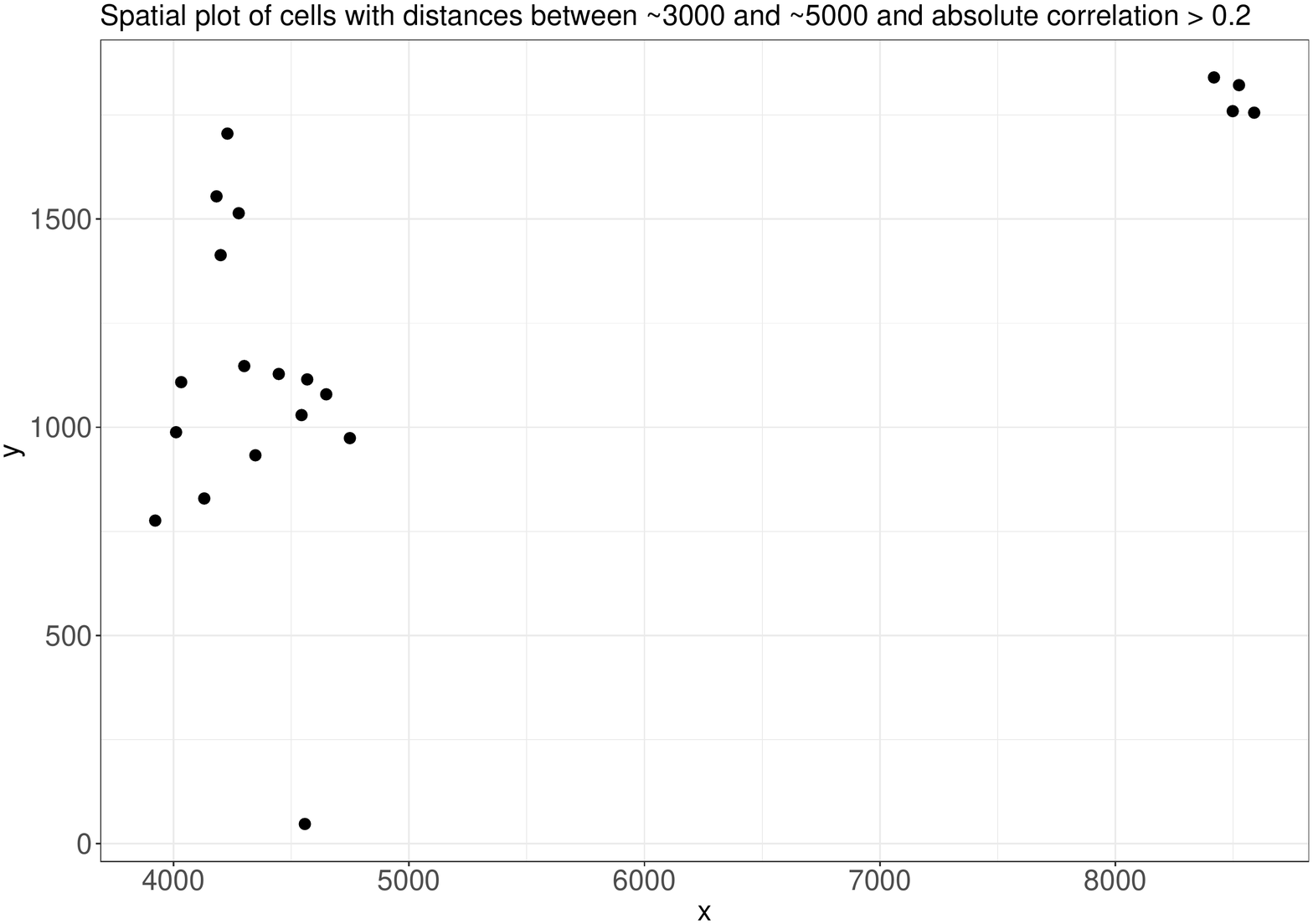}
\caption{Single-sample. Spatial plot of separated cells showing high correlations.}
\label{fig:CellCellConnection}
\end{figure}

\subsection{Multi-sample analysis with comparison}
\label{sec:real_data_multiple_sample_comparison}

We now consider the full STARMap dataset, which consists of data from four independent samples/mice. The four mice were ``dark housed" for four days and then either exposed to light or kept in the dark for another one hour before obtaining their measurements. The number of  cells varied from 931 to 1167 for the four different samples. We selected the first 50 spatially variable genes using R package \texttt{DR.SC} \citep{DR.SC.package} for each of the four independent samples and considered a common set of genes, which led to  $17$ genes. We ran our Gibbs sampler for 5,000 iterations. To monitor the convergence we looked at the traceplot of the log-likelihood, after discarding the first 2,500 samples as burn-in, which showed good mixing (Figure \ref{fig:LL.four}). The corresponding autocorrelation  plot (Figure \ref{fig:ACF.Four}) showed low autocorrelation and hence no thinning was needed. Using the posterior estimate of the spatial correlation matrix, we performed spectral clustering. Treating the spatial correlation matrix for each sample as the weighted similarity matrix, we computed its normalized graph Laplacian matrix. We looked at the plot of the first $10$ smallest eigenvalues of the normalized graph Laplacian matrix and chose the number of eigenvectors $k$, for spectral clustering from the elbow of the plot (Figure \ref{fig:EP}). We extracted the $k$ eigenvectors corresponding to the $k$ smallest eigenvalues and performed $k$-means clustering on the eigenvectors \citep{SC_tutorial}. We chose the number of clusters from the elbow plot of the total within sum of squares (Figure \ref{fig:WSS}). 

\begin{figure}[http]
\centering
\begin{subfigure}{0.65\textwidth}
  \centering
  \includegraphics[width= 1\linewidth]{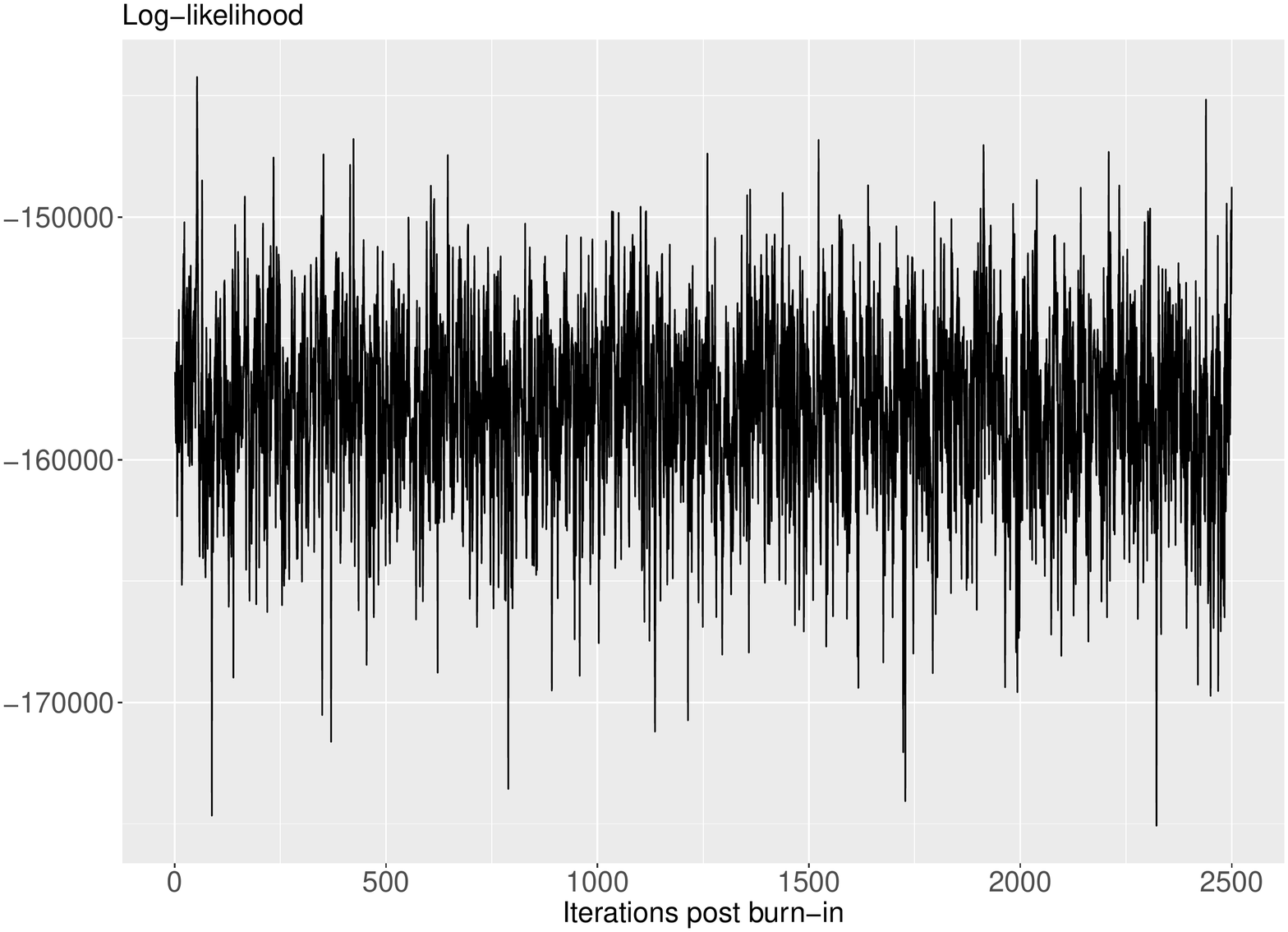}
  \caption{Multi-sample. Traceplot of log-likelihood.}
  \label{fig:LL.four}
\end{subfigure}
\par
\begin{subfigure}{.65\textwidth}
  \centering
 \includegraphics[width= 1\linewidth]{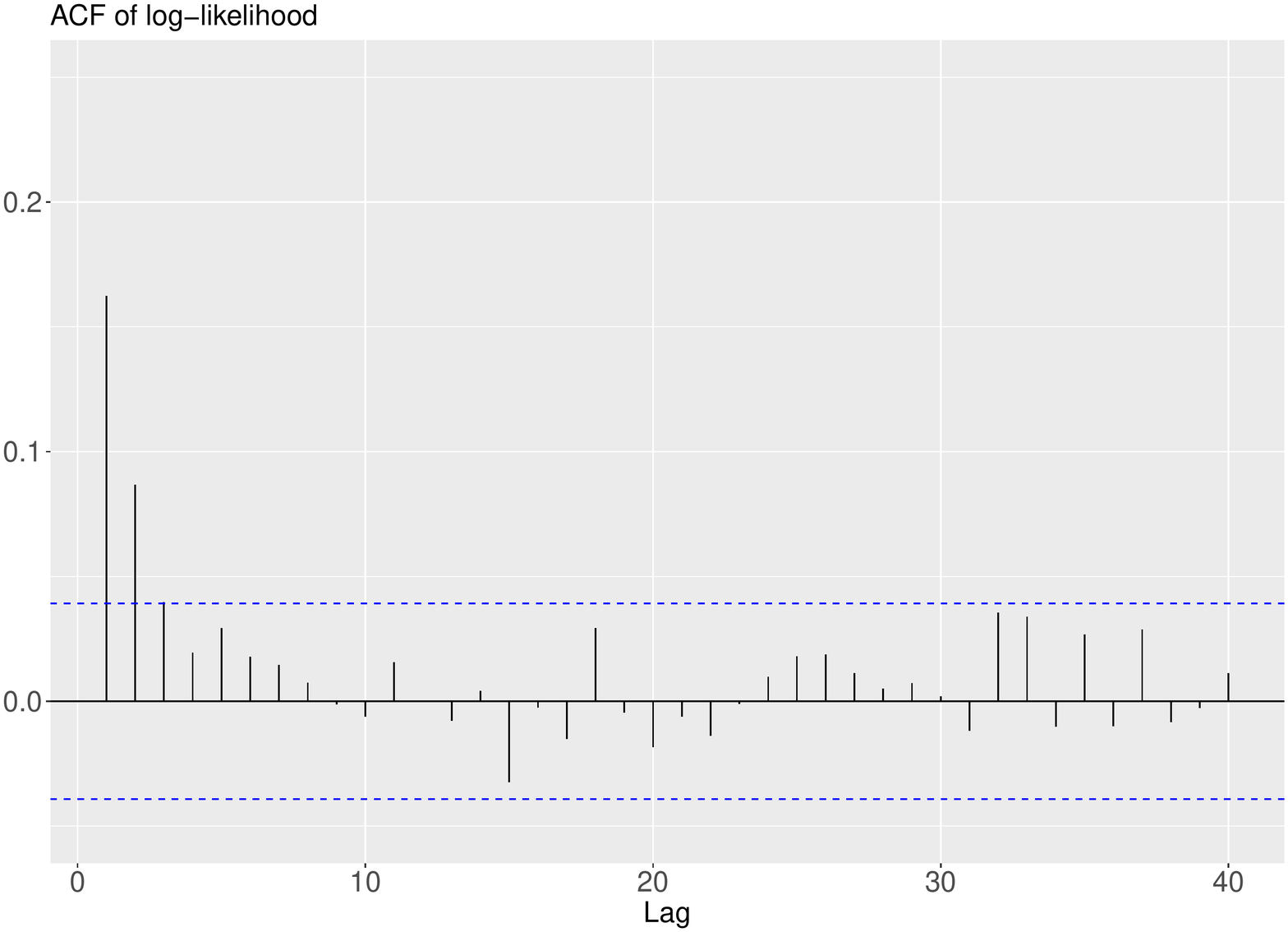}
  \caption{Multi-sample. Autocorrelation plot of log-likelihood.}
  \label{fig:ACF.Four}
\end{subfigure}
\caption{Multi-sample. Traceplot and Autocorrelation plot of log-likelihood.}
\label{fig:LL.ACF.four}
\end{figure}

\begin{figure}
        \centering
        \begin{subfigure}[b]{0.48\textwidth}
            \centering
            \includegraphics[width=\textwidth]{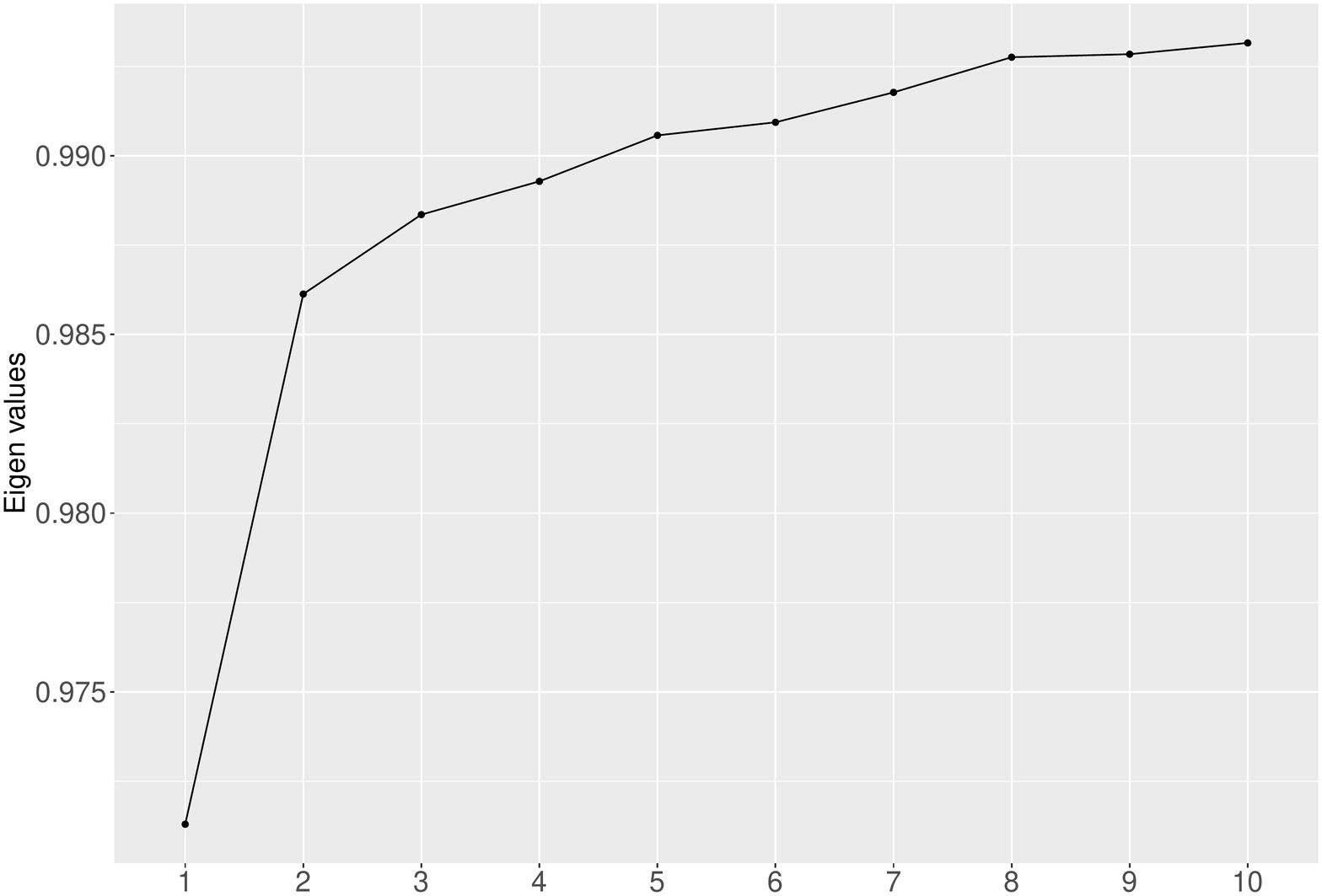}
            \caption{\small Sample 1.}  
            \label{fig:EP1}
        \end{subfigure}
        \hfill
        \begin{subfigure}[b]{0.48\textwidth}  
            \centering 
            \includegraphics[width=\textwidth]{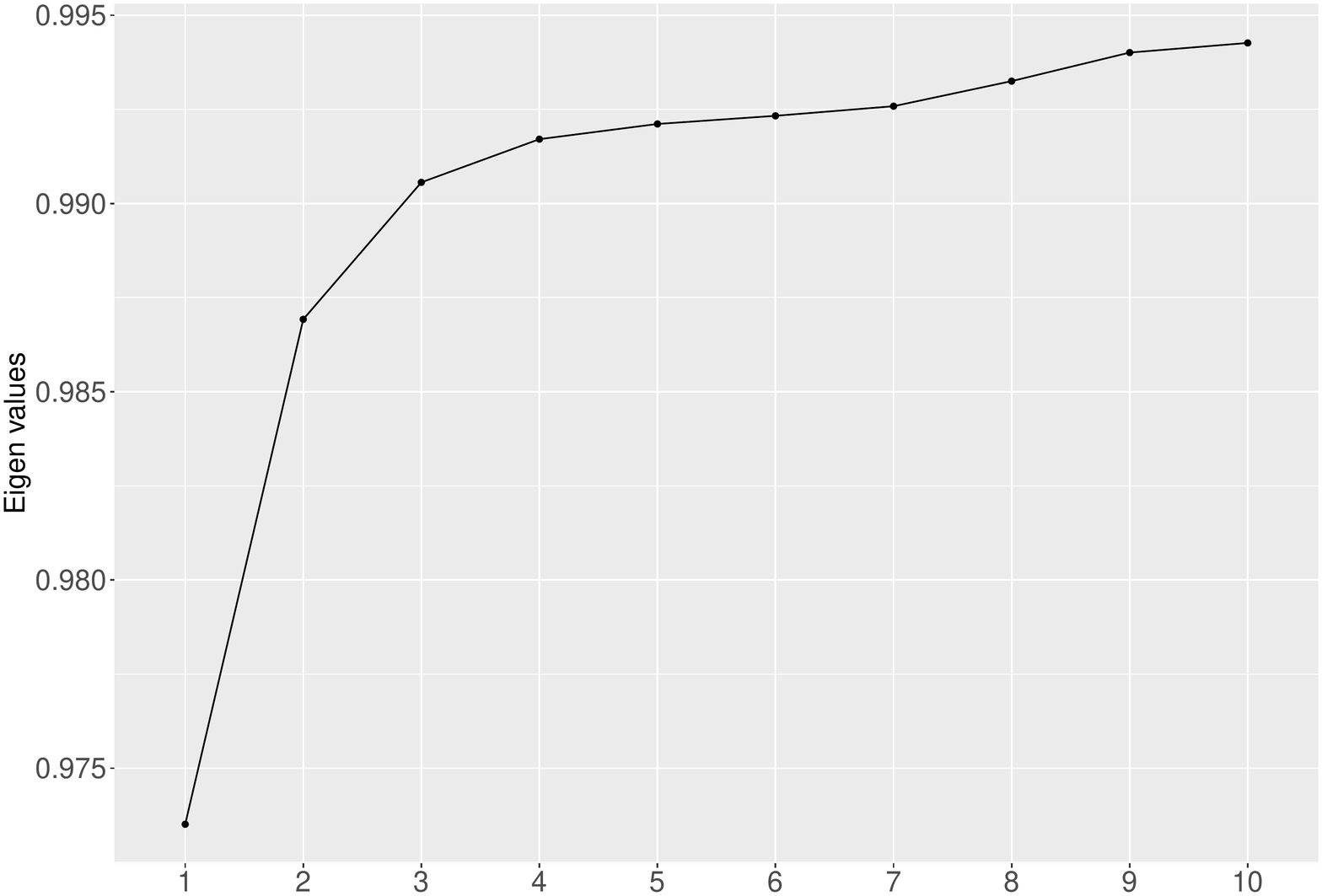}
            \caption{\small Sample 2}    
            \label{fig:EP2}
        \end{subfigure}
        \vskip\baselineskip
        \begin{subfigure}[b]{0.48\textwidth}   
            \centering 
            \includegraphics[width=\textwidth]{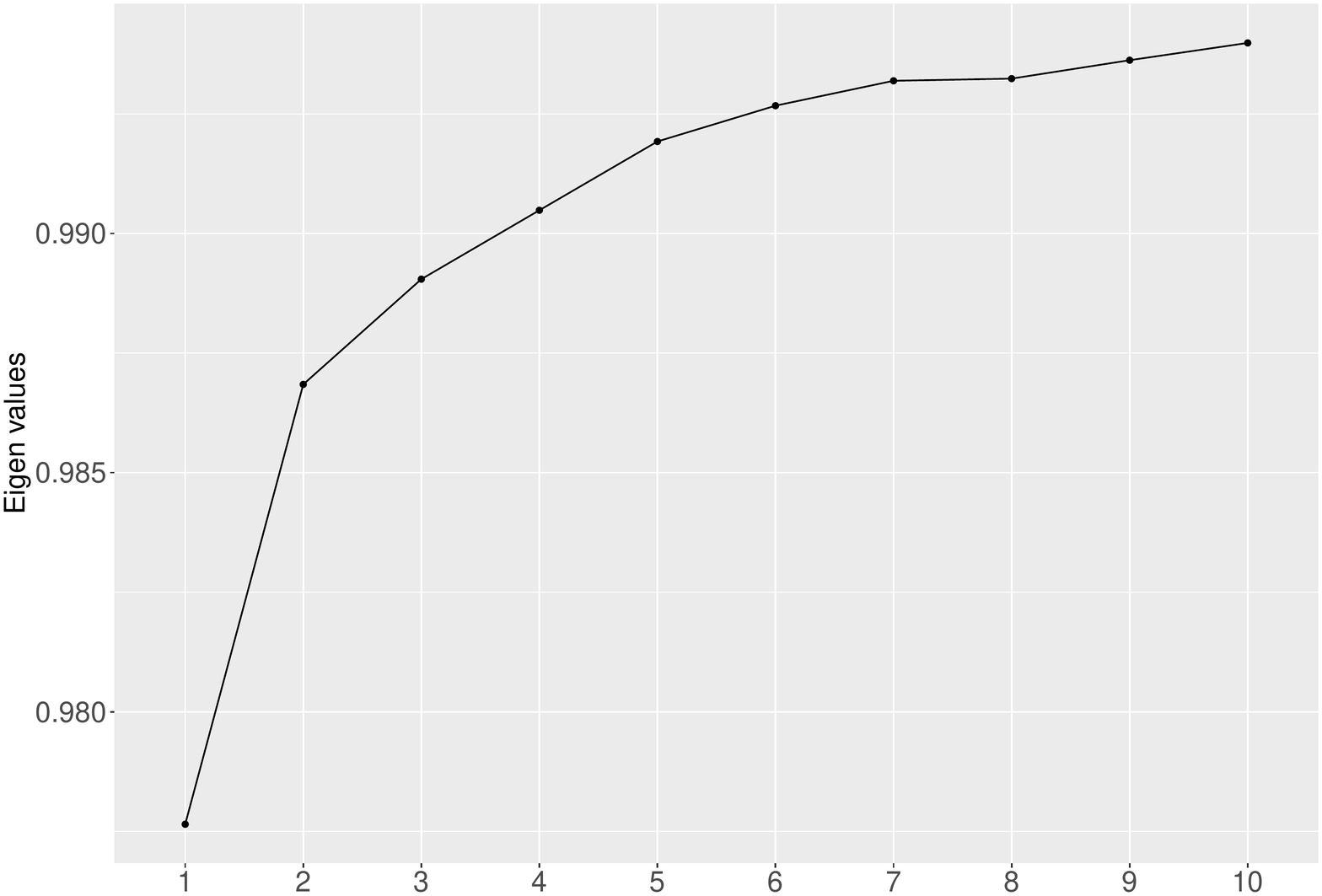}
            \caption{\small Sample 3.} 
            \label{fig:EP3}
        \end{subfigure}
        \hfill
        \begin{subfigure}[b]{0.48\textwidth}   
            \centering 
            \includegraphics[width=\textwidth]{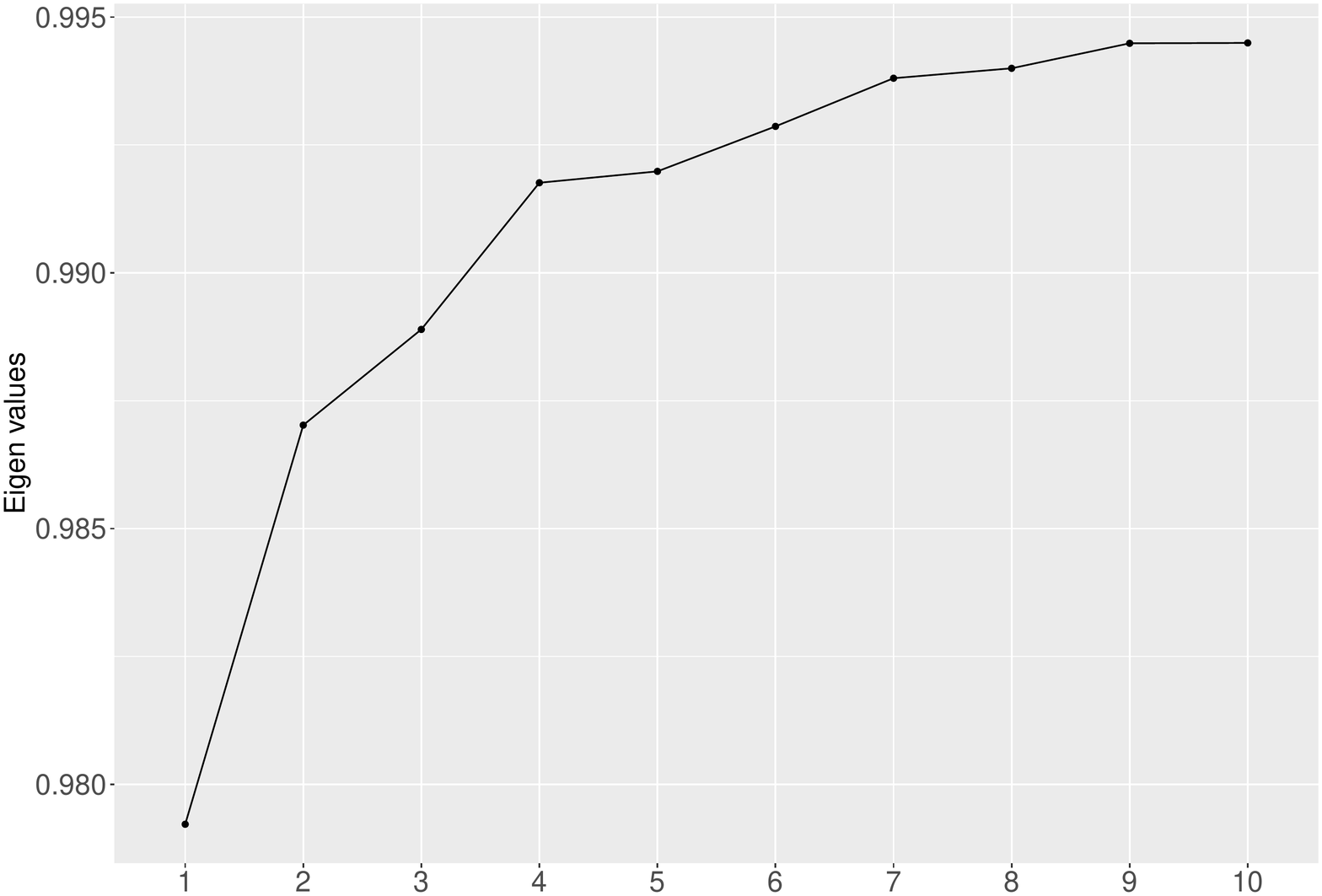}
            \caption{\small Sample 4.}
            \label{fig:EP4}
        \end{subfigure}
        \caption{Multi-sample. The plot of first $10$ smallest eigenvalues of the graph Laplacian matrix for the different samples. We chose the first 3 eigenvectors for samples 1-3 and first 4 eigenvectors for sample 4 to perform spectral clustering.} 
        \label{fig:EP}
    \end{figure}
    
 \begin{figure}
        \centering
        \begin{subfigure}[b]{0.48\textwidth}
            \centering
            \includegraphics[width=\textwidth]{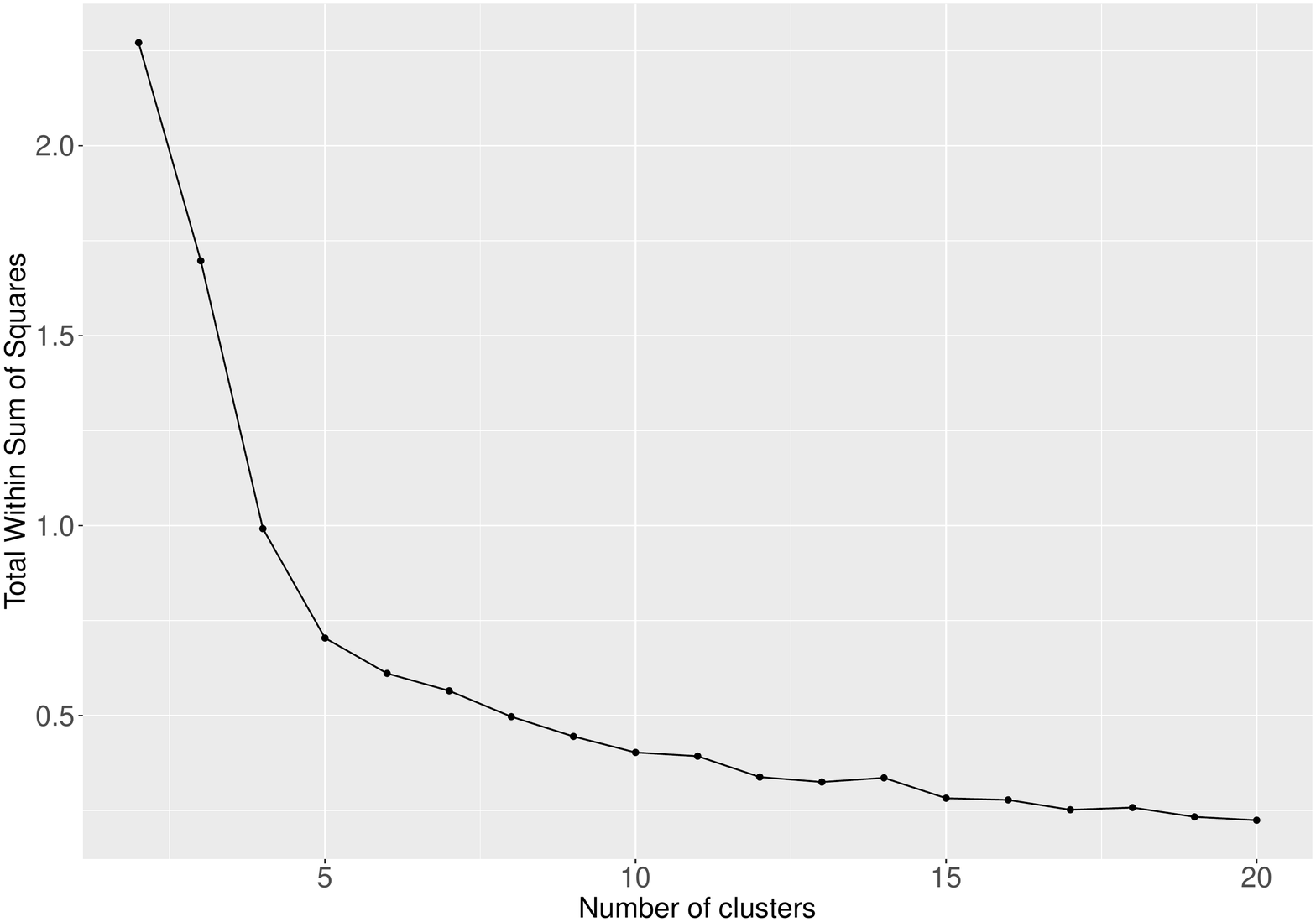}
            \caption{\small Sample 1.}  
            \label{fig:WSS1}
        \end{subfigure}
        \hfill
        \begin{subfigure}[b]{0.48\textwidth}  
            \centering 
            \includegraphics[width=\textwidth]{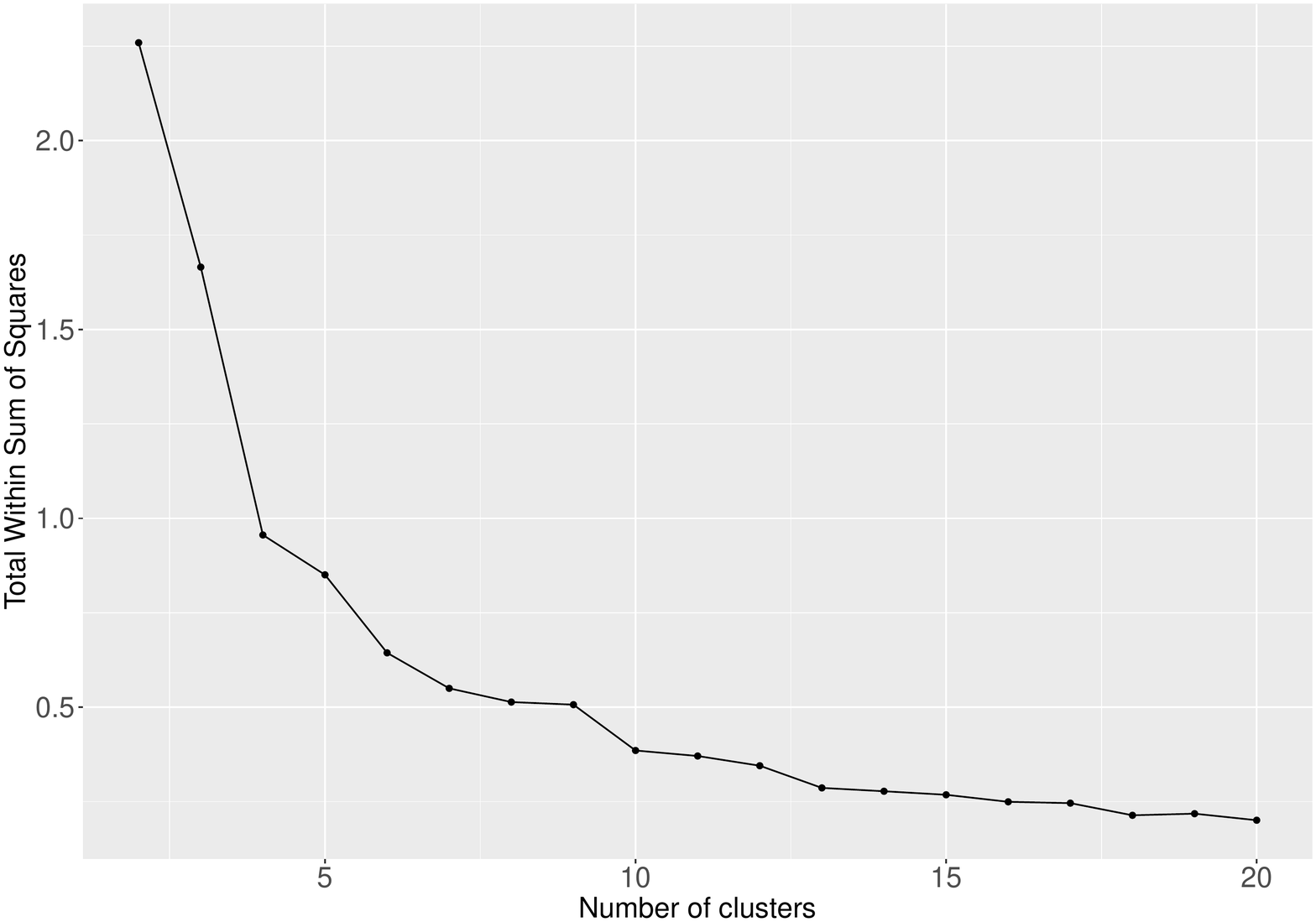}
            \caption{\small Sample 2}    
            \label{fig:WSS2}
        \end{subfigure}
        \vskip\baselineskip
        \begin{subfigure}[b]{0.48\textwidth}   
            \centering 
            \includegraphics[width=\textwidth]{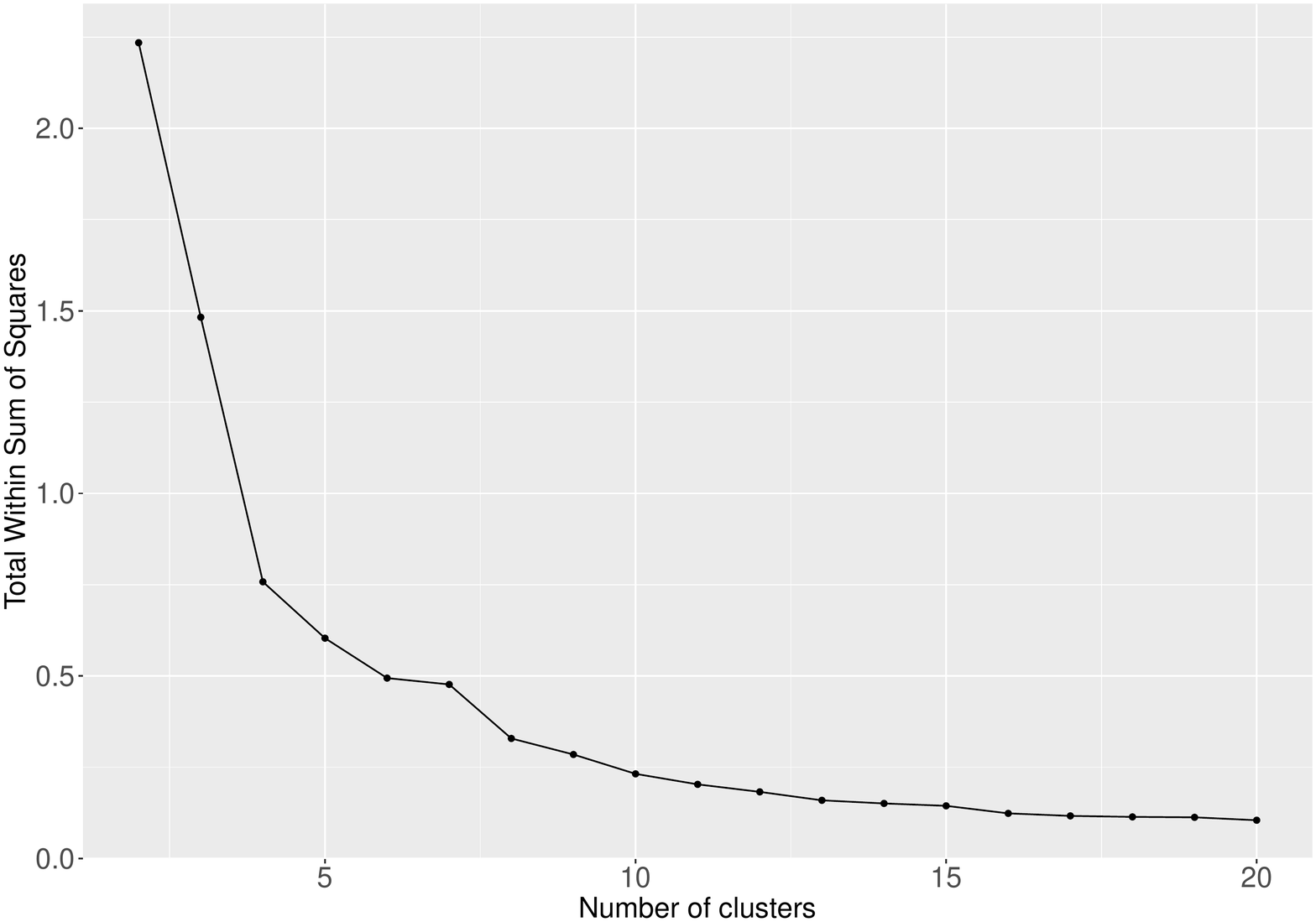}
            \caption{\small Sample 3.} 
            \label{fig:WSS3}
        \end{subfigure}
        \hfill
        \begin{subfigure}[b]{0.48\textwidth}   
            \centering 
            \includegraphics[width=\textwidth]{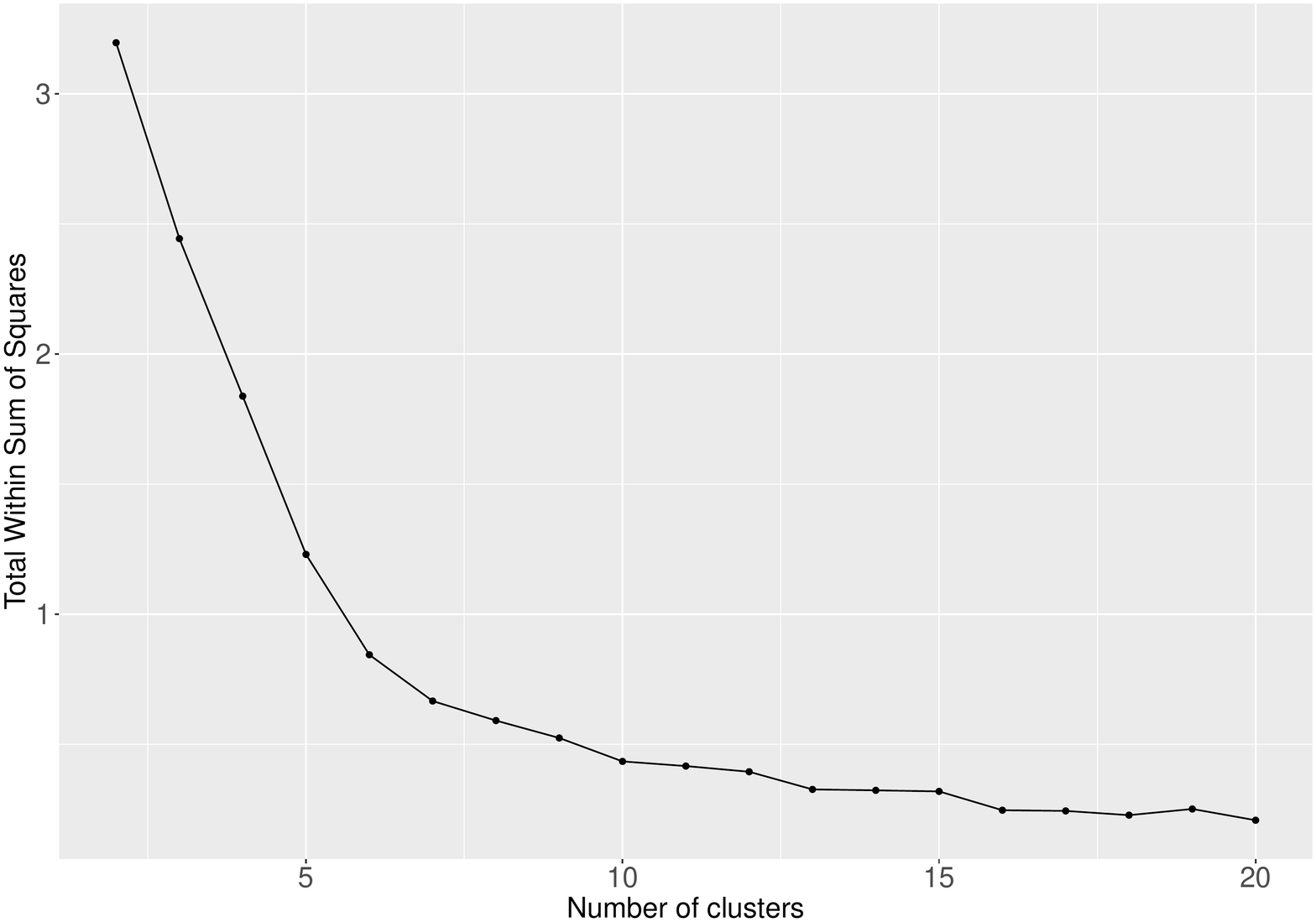}
            \caption{\small Sample 4.}
            \label{fig:WSS4}
        \end{subfigure}
        \caption{Multi-sample. The plot of within cluster sum of squares (WSS) against number of clusters for the different samples. The elbow can be seen to correspond to  $5$ clusters for all of the four samples.} 
        \label{fig:WSS}
    \end{figure}

We compared the proposed method to a recent spatial clustering algorithm specifically designed for spatial transcriptomic data implemented in the \texttt{DR.SC} R package \citep{DR.SC.package}. 

\begin{table}[http]
\centering
\begin{tabular}{ |c|c|c| }
 \hline
  \multicolumn{3}{|c|}{Total Within Cluster Sum of Squares} \\ 
 \cline{1-3}
  Sample & Using our Posterior estimate & Using R package \texttt{DR.SC} \\
  \hline
 1 &  13969 &   16262 \\
 2 &  24718 &   27247 \\
 3 &  12185 &   13716 \\
 4 &  13191 &   14517 \\
 \hline
\end{tabular}
\caption{Multi-sample. The total within cluster sum of squares for the two methods for different samples.}
\label{tab:WCSS_replicated}
\end{table}
Figure \ref{fig:spatial clustering with relicated data}  gives the plot of spatial clustering using the proposed method and the \texttt{DR.SC} R package for all four samples. 
Although we do not have the ground truth, the clusters identified by the proposed method exhibits interesting spatial patterns, which are much less clear under \texttt{DR.SC}. Lack of spatial patterns is biologically less plausible.
The total within cluster sum of squares (WCSS) of the gene expression (Table \ref{tab:WCSS_replicated}) also indicate that the cells  are molecularly more similar within each of the clusters estimated from the proposed method than \texttt{DR.SC}. 
The posterior estimate of the row correlation matrix was used to construct the gene co-expression network. As before, considering a cut-off of 0.1 for the partial correlations, Figure \ref{fig:GRN_replicated} shows the estimated network. We can see that the gene ``Arx" is a hub and shows significant association with multiple genes. The importance of this gene has been established. The ``Arx" gene provides instructions for producing a protein that regulates the activity of other genes (\citealp{Arx1, Arx2}). 
An independent study \citep{BrainAtlasPaper, database1} reveals low expression of the gene ``Arx" and ``Slc17a7", both of which are hub genes in our estimated network and are connected to each other. It is also revealed that the genes ``Otof" and ``Bcl6" show low expression. A possible explanation for this phenomenon is their connection with the hub gene ``Arx" (Figure \ref{fig:GRN_replicated}), which possibly down regulates the expression of the ``Otof" and ``Bcl6" genes as can be seen from their positive partial correlations with the ``Arx" gene in Figure \ref{fig:CorrPlotSTARmap}. 
\begin{figure}[http]
\centering
\begin{subfigure}[t]{0.65\textwidth}
  \centering
  \includegraphics[width= \linewidth]{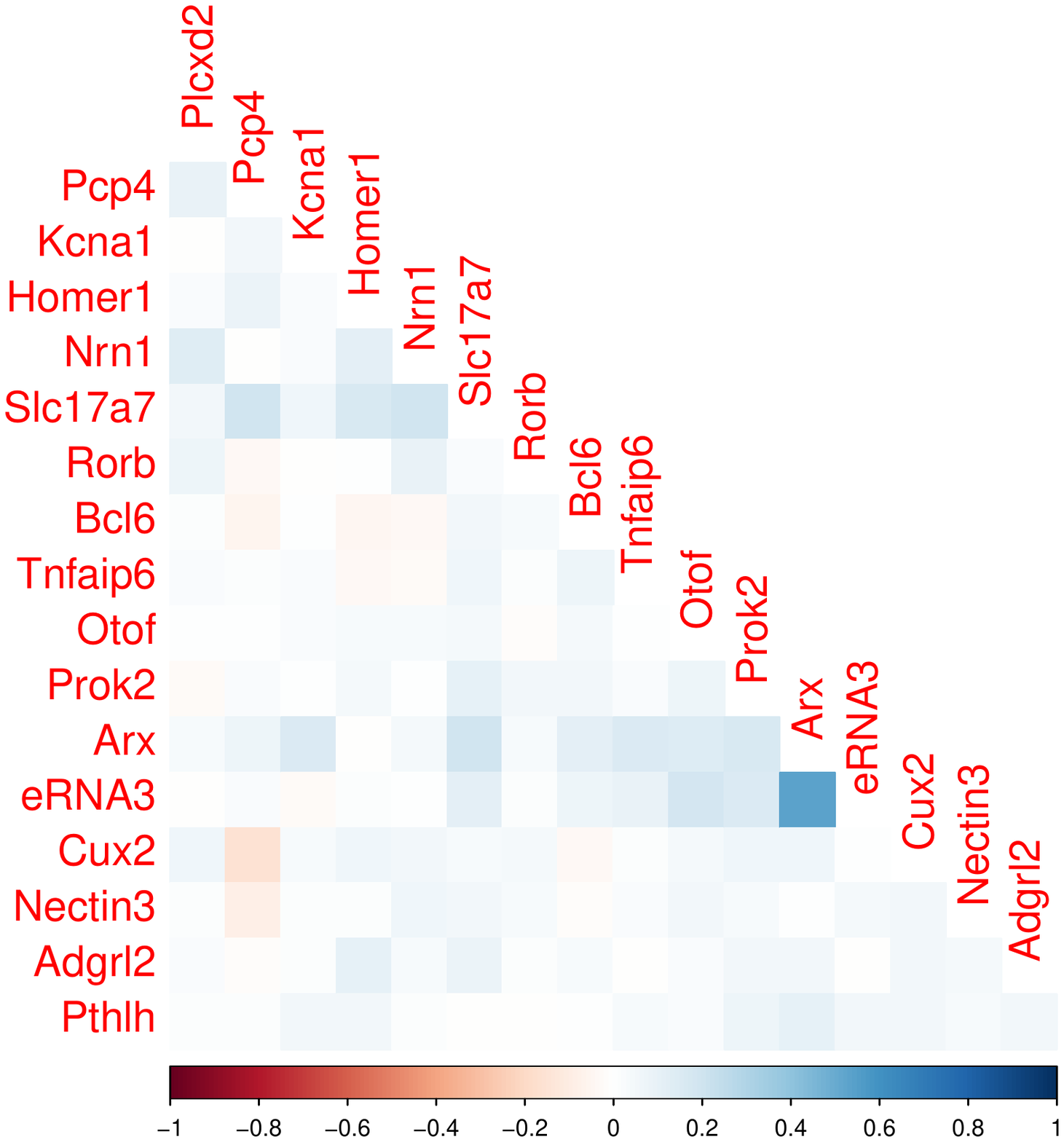}
  \caption{Heatmap of the estimate of partial correlations between the selected genes.}
  \label{fig:CorrPlotSTARmap}
\end{subfigure}
\qquad
\begin{subfigure}[t]{.65\textwidth}
  \centering
  \includegraphics[width= \linewidth]{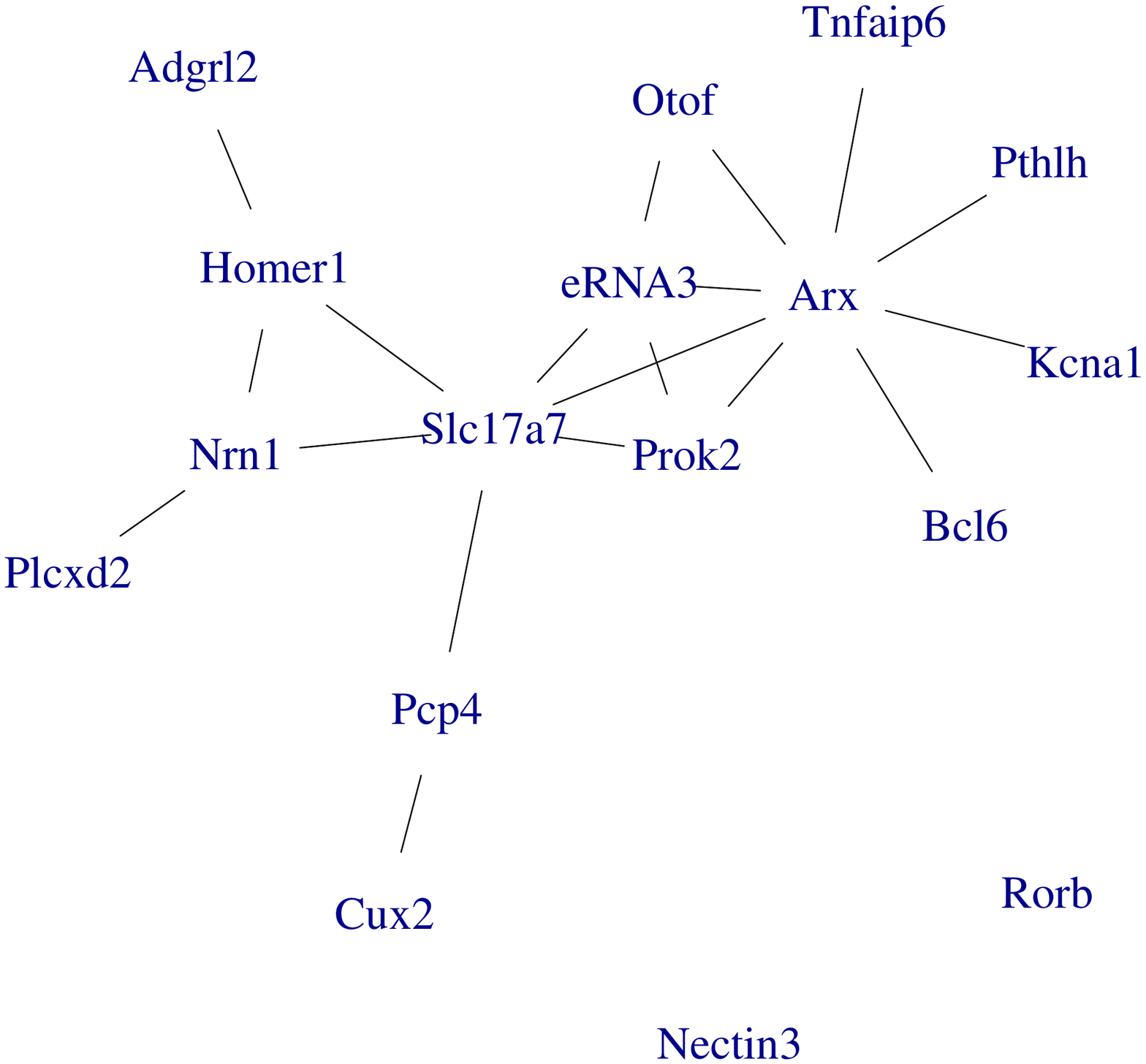}
  \caption{Estimated gene network from our posterior row precision matrix between the selected genes.}
  \label{fig:GRN_replicated}
\end{subfigure}
\caption{Multi-sample analysis.}
\label{fig:gene_corr_network_replicated}
\end{figure}

\begin{figure}[http]
\centering
\begin{subfigure}[t]{0.48\textwidth}
  \centering
  \includegraphics[width= \linewidth]{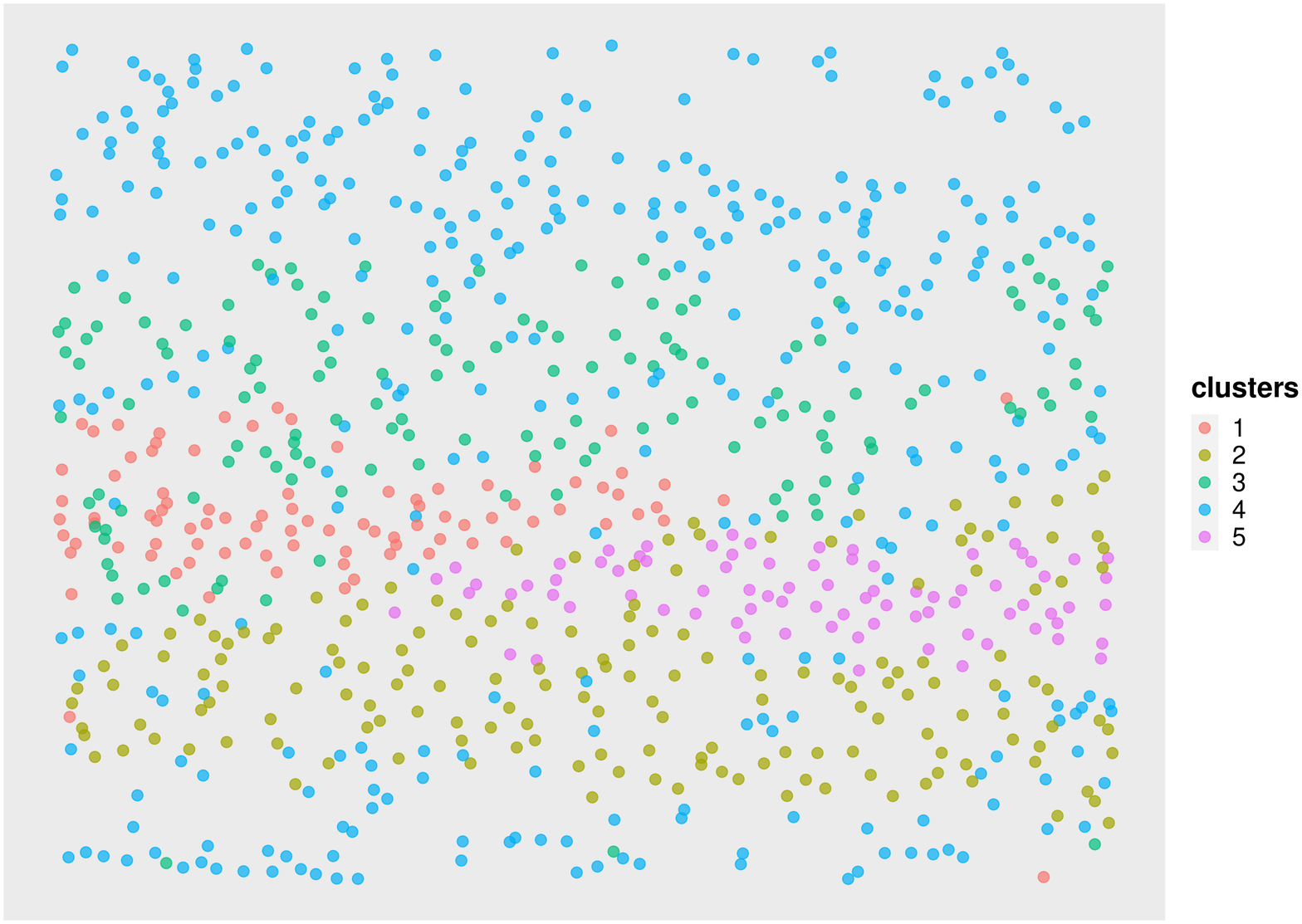}
  \caption{Clustering with spectral clustering with posterior spatial correlations for sample 1.}
  \label{fig:clustering_bayes1}
\end{subfigure}
\begin{subfigure}[t]{.48\textwidth}
  \centering
  \includegraphics[width=\linewidth]{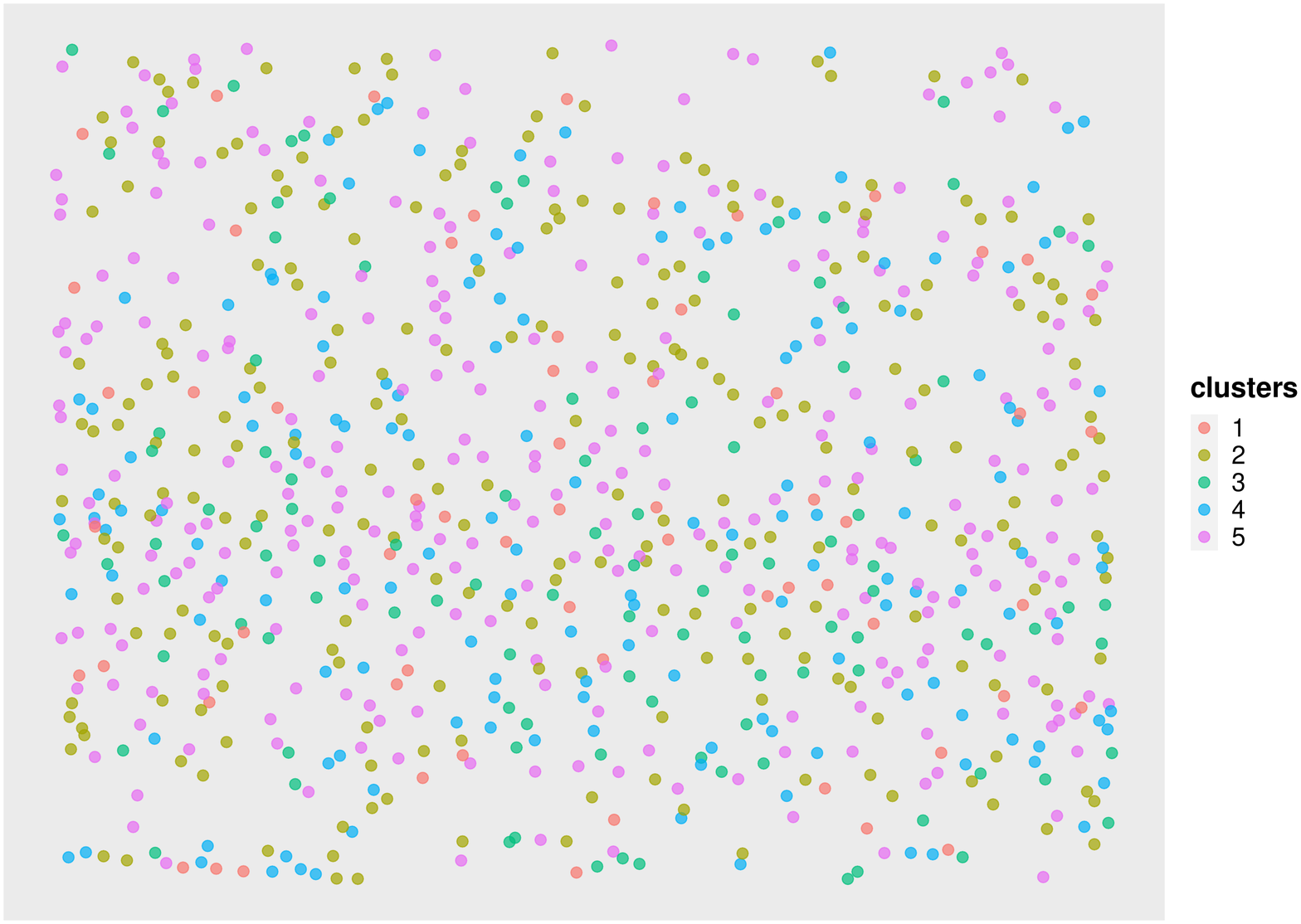}
  \caption{Clustering using \texttt{DR.SC} R package for sample 1.}
  \label{fig:clustering_drsc1}
\end{subfigure}
\begin{subfigure}[t]{.48\textwidth}
  \centering
  \includegraphics[width=\linewidth]{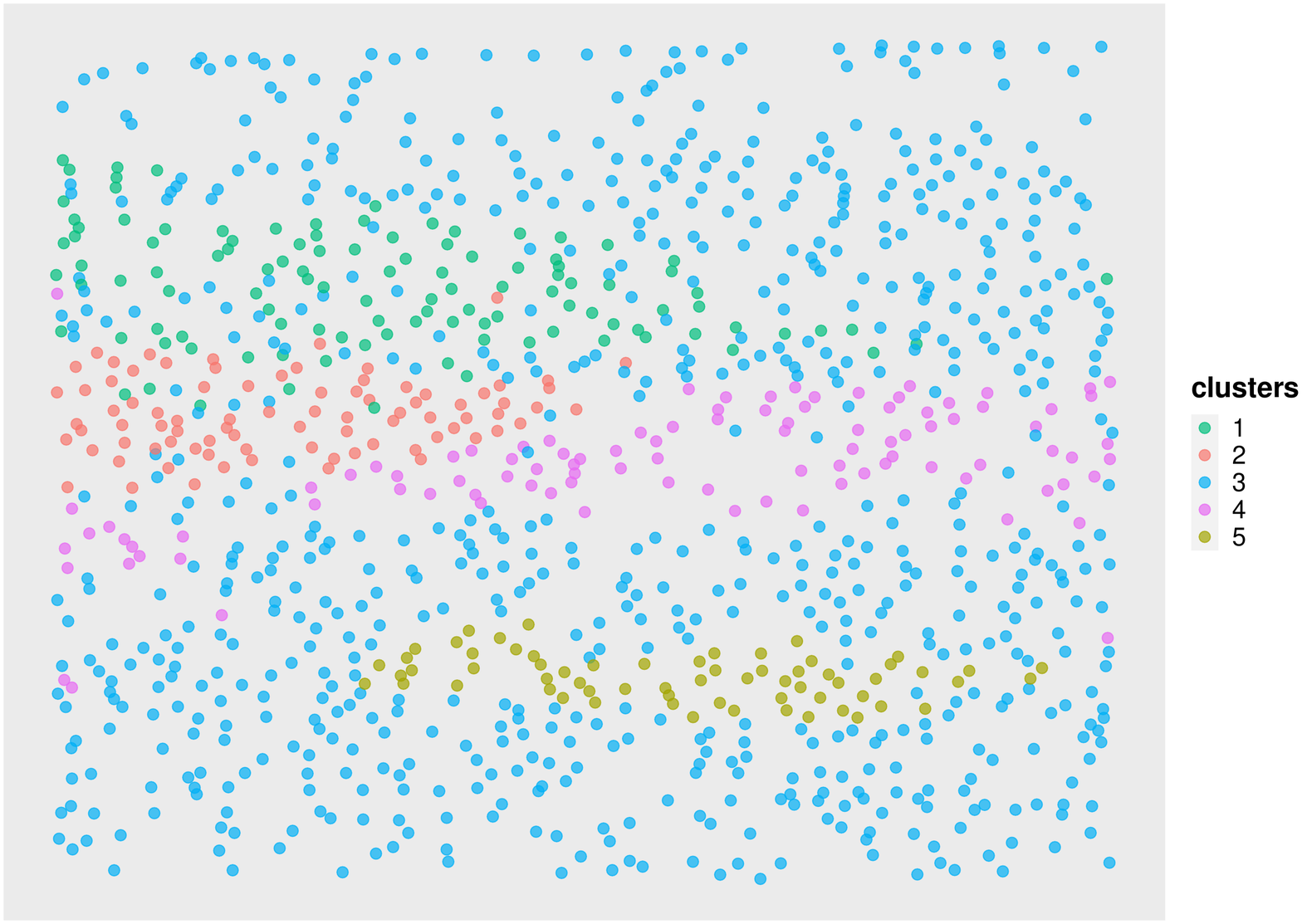}
  \caption{Clustering with spectral clustering with posterior spatial correlations for sample 2.}
  \label{fig:clustering_bayes2}
\end{subfigure}
\begin{subfigure}[t]{.48\textwidth}
  \centering
  \includegraphics[width=\linewidth]{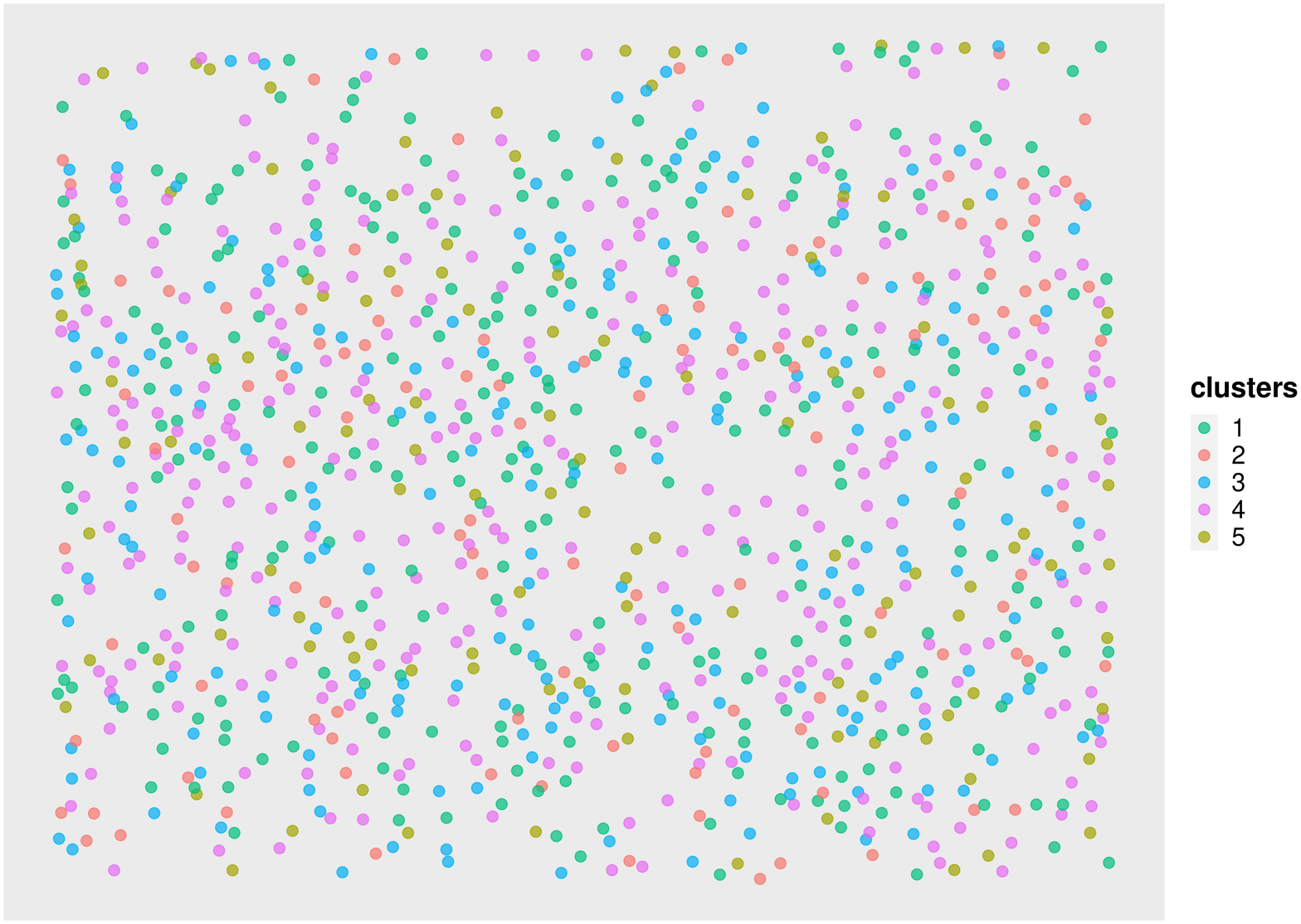}
  \caption{Clustering using \texttt{DR.SC} R package for sample 2.}
  \label{fig:clustering_drsc2}
\end{subfigure}
\begin{subfigure}[t]{.48\textwidth}
  \centering
  \includegraphics[width=\linewidth]{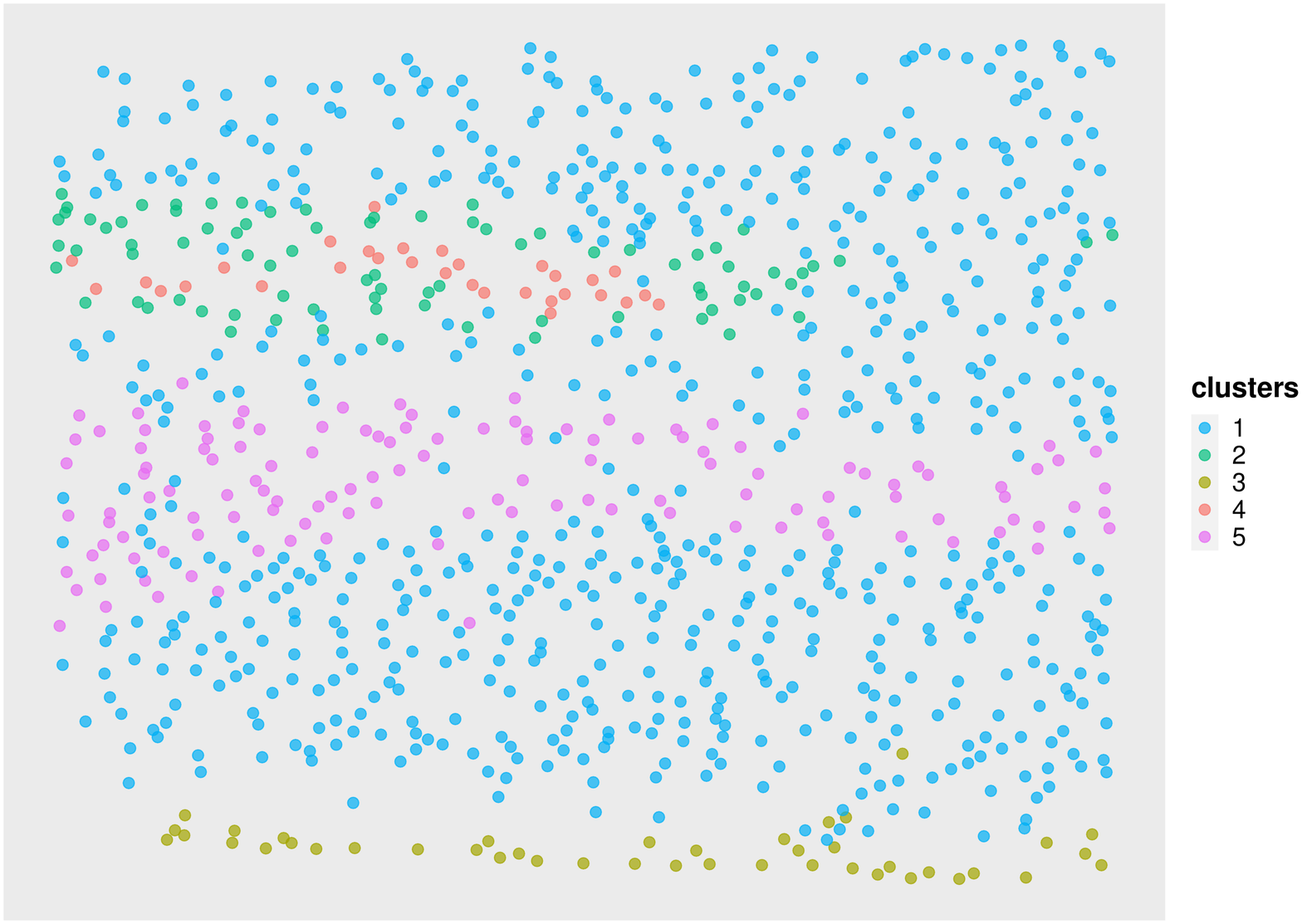}
  \caption{Clustering with spectral clustering with posterior spatial correlations for sample 3.}
  \label{fig:clustering_bayes3}
\end{subfigure}
\begin{subfigure}[t]{.48\textwidth}
  \centering
  \includegraphics[width=\linewidth]{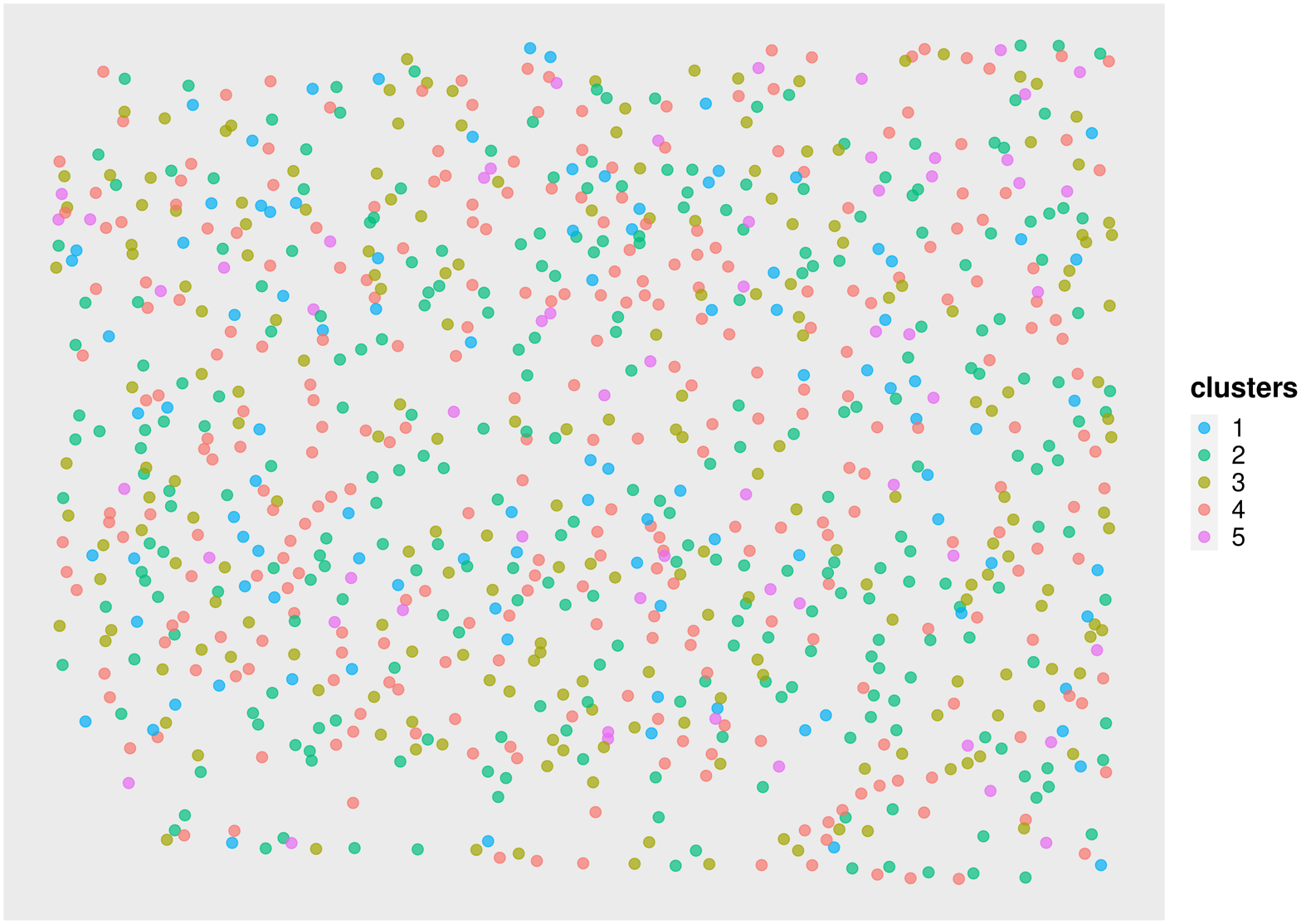}
  \caption{Clustering using \texttt{DR.SC} R package for sample 3.}
  \label{fig:clustering_drsc3}
\end{subfigure}
\end{figure}
\begin{figure}[http]
\ContinuedFloat
\centering
\begin{subfigure}[t]{.48\textwidth}
  \centering
  \includegraphics[width=\linewidth]{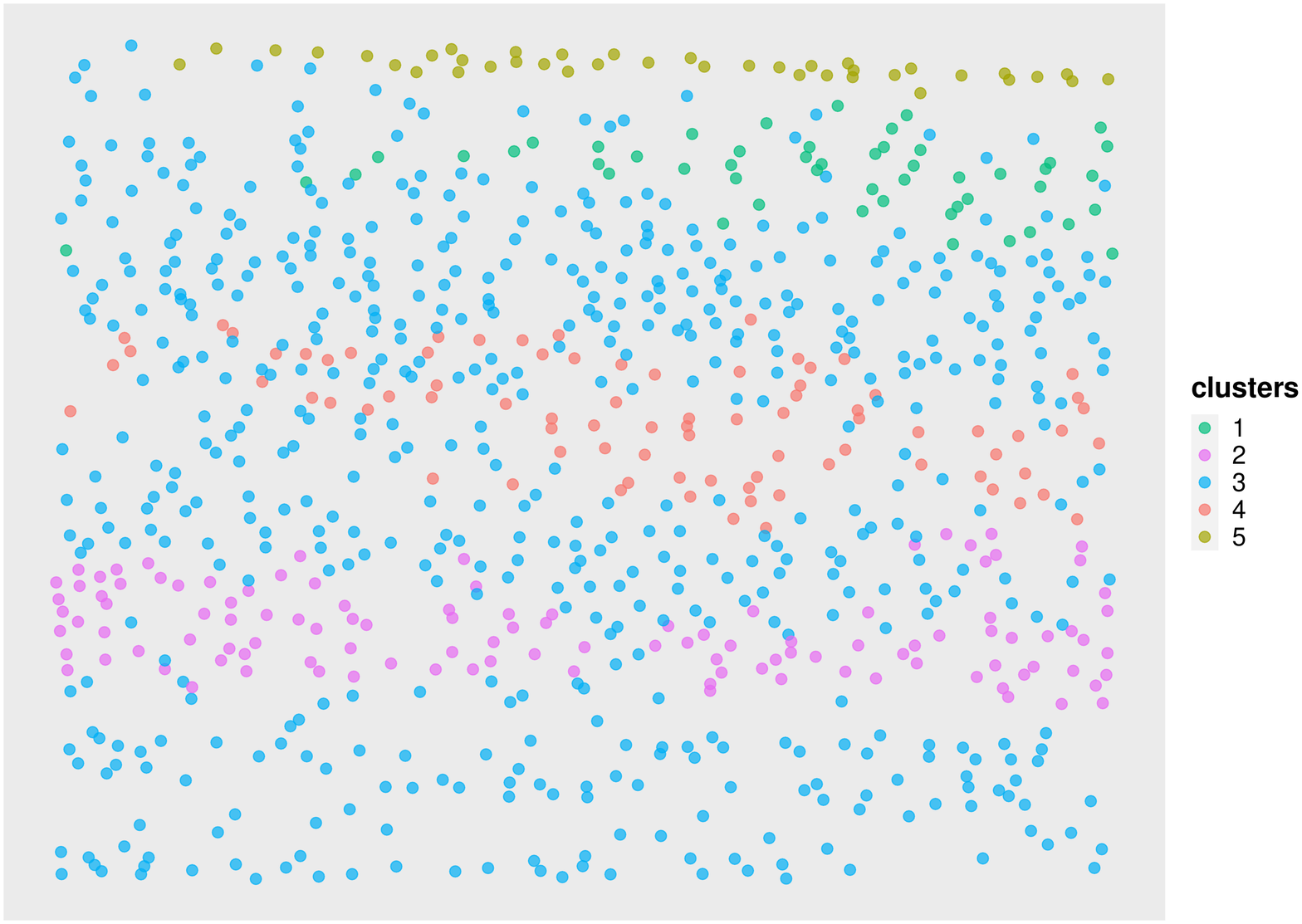}
  \caption{Clustering with spectral clustering with posterior spatial correlations for sample 4.}
  \label{fig:clustering_bayes4}
\end{subfigure}
\begin{subfigure}[t]{.48\textwidth}
  \centering
  \includegraphics[width=\linewidth]{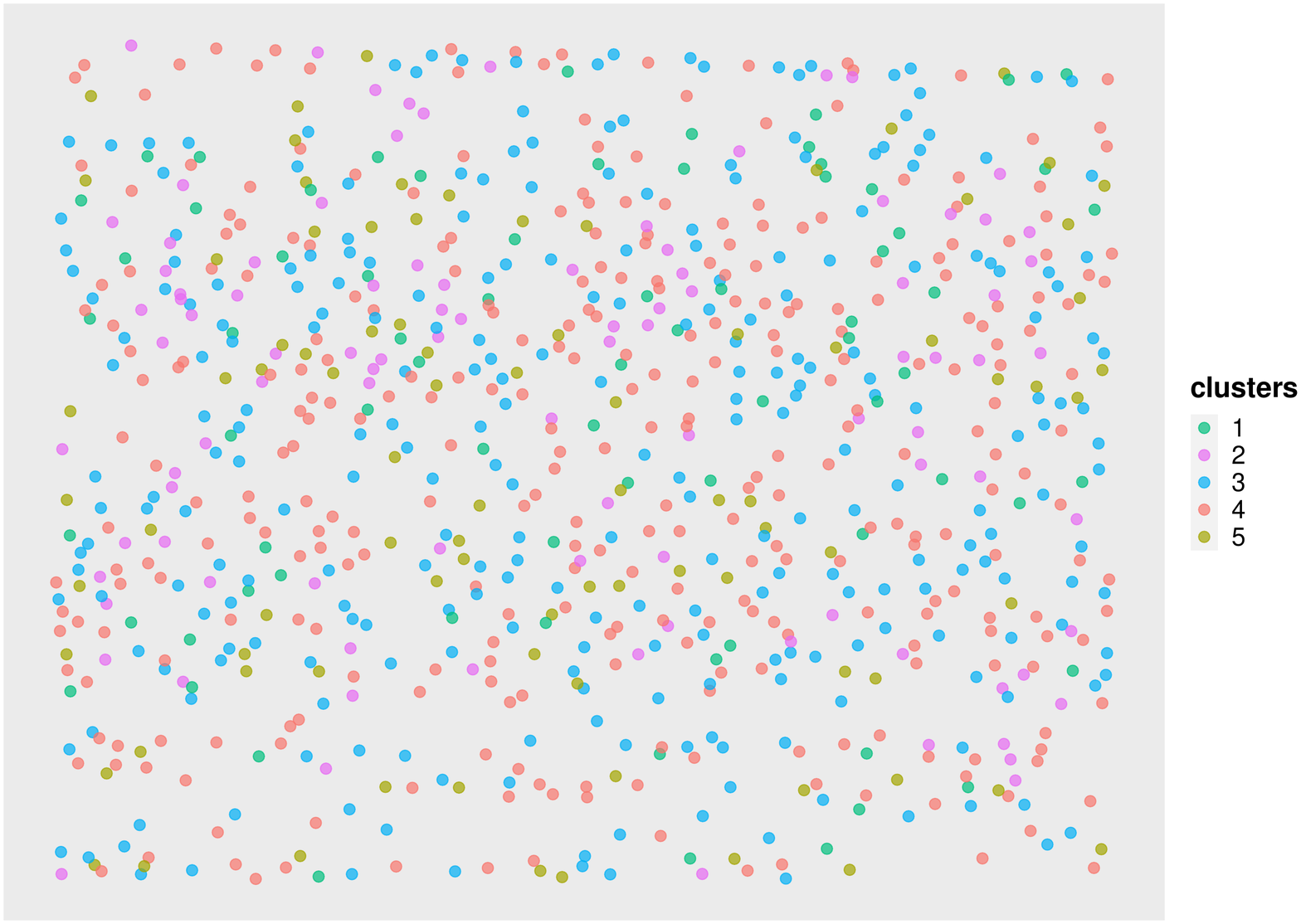}
  \caption{Clustering using \texttt{DR.SC} R package for sample 4.}
  \label{fig:clustering_drsc4}
\end{subfigure}
\caption{Multi-sample. Spatial clustering from the proposed method compared with the state of the art method as implemented by \texttt{DR.SC} R package.}
\label{fig:spatial clustering with relicated data}
\end{figure}

\section{Discussion}\label{sec:discussion}
We have introduced a Bayesian nonparametric method of estimation of covariance matrices for matrix-variate data wherein both the rows and columns of the matrix-variate data are correlated by the very design of the study. We have further extended our method to the case where we have multiple independent samples of the matrix-variate data observed over a possibly different set of spatial locations but a common set of genes making up the rows of the matrices. We have illustrated the power of our method using simulations and real data where we made comparison  with existing methods. \par
There are a few possible future directions for this work. First, it may be possible to consider spatial transcriptomic studies with large number of observed genes. The challenge is to define the joint distribution over the matrix-variate data, which along with the estimation of covariance matrices would allow for automatic selection of relevant genes from the entire gene set through some Bayesian variable selection criterion. Second, it may be possible to incorporate some aspects of graphical models to understand further the interaction between the genes. Third, it may be possible to speed up computations using variational inference methods. Fourth, it may also be possible to consider spatial modelling of the raw gene expression data using multivariate zero-inflated count distributions.
\clearpage
\bibliographystyle{apalike}
\bibliography{references}

\end{document}